\documentclass[11pt]{article}

\usepackage[margin=1in]{geometry}
\usepackage{amsthm,amsmath,amsfonts,amssymb}
\usepackage[percent]{overpic}

\usepackage{colonequals}
\usepackage{braket}

\usepackage{hyperref}
\hypersetup{
    bookmarksnumbered=true, 
    unicode=false, 
    pdfstartview={FitH}, 
    pdftitle={Quantum query complexity of minor-closed graph properties}, 
    pdfauthor={Andrew M. Childs and Robin Kothari}, 
    pdfsubject={}, 
    pdfcreator={}, 
    pdfproducer={}, 
    pdfkeywords={}, 
    pdfnewwindow=true, 
    colorlinks=true, 
    linkcolor=blue, 
    citecolor=blue, 
    filecolor=blue, 
    urlcolor=blue 
}
\usepackage{hypcap}

\newtheorem{theorem}{Theorem}[section]
\newtheorem{lemma}{Lemma}[section]
\newtheorem{proposition}{Proposition}[section]
\newtheorem{corollary}{Corollary}[section]

\newcommand{\eq}[1]{\hyperref[eq:#1]{(\ref*{eq:#1})}}
\renewcommand{\sec}[1]{\hyperref[sec:#1]{Section~\ref*{sec:#1}}}
\newcommand{\thm}[1]{\hyperref[thm:#1]{Theorem~\ref*{thm:#1}}}
\newcommand{\lem}[1]{\hyperref[lem:#1]{Lemma~\ref*{lem:#1}}}
\newcommand{\prop}[1]{\hyperref[prop:#1]{Proposition~\ref*{prop:#1}}}
\newcommand{\cor}[1]{\hyperref[cor:#1]{Corollary~\ref*{cor:#1}}}
\newcommand{\fig}[1]{\hyperref[fig:#1]{Figure~\ref*{fig:#1}}}
\newcommand{\tab}[1]{\hyperref[tab:#1]{Table~\ref*{tab:#1}}}

\newcommand{\comment}[1]{}

\newcommand{\Z}{{\mathbb{Z}}}
\renewcommand{\P}{\mathcal{P}}

\newcommand{\poly}{\mathop{\mathrm{poly}}}

\newcommand{\vc}{\mathop{\mathrm{vc}}}
\renewcommand{\d}{\mathrm{d}}

\renewcommand{\th}[1]{${#1}^{\textrm{th}}$}

\newcommand{\ceil}[1]{\lceil{#1}\rceil}

\newcommand{\eqrange}[2]{(\ref{eq:#1}--\ref{eq:#2})}

\renewcommand{\(}{\left(}
\renewcommand{\)}{\right)}
\newcommand{\defeq}{\colonequals}

\begin{document}


\title{Quantum query complexity of minor-closed graph properties\thanks{This paper is an expanded version of an extended abstract that appeared in the proceedings of STACS 2011 \cite{CK11}.}}

\author{
\normalsize Andrew M.\ Childs\thanks{amchilds@uwaterloo.ca} \\[.5ex]
\small Department of Combinatorics \& Optimization \\
\small and Institute for Quantum Computing \\
\small University of Waterloo
\and
\normalsize Robin Kothari\thanks{rkothari@cs.uwaterloo.ca} \\[.5ex]
\small David R.\ Cheriton School of Computer Science \\
\small and Institute for Quantum Computing \\
\small University of Waterloo
}

\date{}
\maketitle

\renewcommand{\arraystretch}{1.6} 

\begin{abstract}
We study the quantum query complexity of minor-closed graph properties, which include such problems as determining whether an $n$-vertex graph is planar, is a forest, or does not contain a path of a given length.
We show that most minor-closed properties---those that cannot be characterized by a finite set of forbidden subgraphs---have quantum query complexity $\Theta(n^{3/2})$.  To establish this, we prove an adversary lower bound using a detailed analysis of the structure of minor-closed properties with respect to forbidden topological minors and forbidden subgraphs.
On the other hand, we show that minor-closed properties (and more generally, sparse graph properties) that can be characterized by finitely many forbidden subgraphs can be solved strictly faster, in $o(n^{3/2})$ queries.  Our algorithms are a novel application of the quantum walk search framework and give improved upper bounds for several subgraph-finding problems.
\end{abstract}

\tableofcontents
\clearpage
\section{Introduction}
\label{sec:intro}

The decision tree model is a simple model of computation for which we can prove good upper and lower bounds. Informally, decision tree complexity, also known as query complexity, counts the number of input bits that must be examined by an algorithm to evaluate a function. In this paper, we focus on the query complexity of deciding whether a graph has a given property.  The query complexity of graph properties has been studied for almost 40 years, yet old and easy-to-state conjectures regarding the deterministic and randomized query complexities of graph properties~\cite{CK07,KSS83,RV75,Ros73} remain unresolved.

The study of query complexity has also been quite fruitful for quantum algorithms. For example, Grover's search algorithm \cite{Gro97} operates in the query model, and Shor's factoring algorithm \cite{Sho97} is based on the solution of a query problem.
However, the quantum query complexity of graph properties can be harder to pin down than its classical counterparts. For monotone graph properties, a wide class of graph properties including almost all the properties considered in this paper, the widely-believed Aanderaa--Karp--Rosenberg conjecture states that the deterministic and randomized query complexities are $\Theta(n^2)$, where $n$ is the number of vertices. On the other hand, there exist monotone graph properties whose quantum query complexity is $\Theta(n)$, and others with quantum query complexity $\Theta(n^2)$. In fact, one can construct a monotone graph property with quantum query complexity $\Theta(n^{1+\alpha})$ for any fixed $0\leq \alpha \leq 1$ using known bounds for the threshold function~\cite{BBC+01}. 

The quantum query complexity of several specific graph properties has been established in prior work.  D{\"u}rr, Heiligman, H{\o}yer, and Mhalla~\cite{DHHM06} studied the query complexity of several graph problems, and showed in particular that connectivity has quantum query complexity $\Theta(n^{3/2})$. Zhang~\cite{Zha05} showed that the quantum query complexity of bipartiteness is $\Theta(n^{3/2})$. Ambainis et al.\ \cite{AIN+08} showed that planarity also has quantum query complexity $\Theta(n^{3/2})$. Berzina et al.~\cite{BDF+03} showed several quantum lower bounds on graph properties, including Hamiltonicity. Sun, Yao, and Zhang studied some non-monotone graph properties~\cite{SYZ04}.

Despite this work, the quantum query complexity of many interesting graph properties remains unresolved.  A well-studied graph property whose query complexity is unknown is the property of containing a triangle (i.e., a cycle on 3 vertices) as a subgraph.  This question was first studied by Buhrman et al.~\cite{BDH+05}, who gave an $O(n+\sqrt{nm})$ query algorithm for graphs with $n$ vertices and $m$ edges. With $m=\Theta(n^2)$, this approach uses  $O(n^{3/2})$ queries, which matches the performance of the simple algorithm that searches for a triangle over the potential $n \choose 3$ triplets of vertices. This was later improved by Magniez, Santha, and Szegedy~\cite{MSS07} to $\tilde O(n^{1.3})$, and then by Magniez, Nayak, Roland, and Santha~\cite{MNRS07} to $O(n^{1.3})$, which is currently the best known algorithm.  However, the best known lower bound for the triangle problem is only $\Omega(n)$ (by a simple reduction from the search problem).  This is partly because one of the main lower bound techniques, the quantum adversary method of Ambainis~\cite{Amb02}, cannot prove a better lower bound due to the certificate complexity barrier~\cite{SS06,Zha05}.

More generally, we can consider the $H$-subgraph containment problem, in which the task is to determine whether the input graph contains a fixed graph $H$ as a subgraph.  Magniez et al.\ also gave a general algorithm for $H$-subgraph containment using $\tilde O(n^{2-2/d})$ queries, where $d>3$ is the number of vertices in $H$ \cite{MSS07}.  Again, the best lower bound known for $H$-subgraph containment is only $\Omega(n)$.

In this paper we study the quantum query complexity of minor-closed graph properties. A property is minor closed if all minors of a graph possessing the property also possess the property. (Graph minors are defined in \sec{prelim}.) Since minor-closed properties can be characterized by forbidden minors, this can be viewed as a variant of subgraph containment in which we look for a given graph as a minor instead of as a subgraph. The canonical example of a minor-closed property is the property of being planar. Other examples include the property of being a forest, being embeddable on a fixed two-dimensional manifold, having treewidth at most $k$, and not containing a path of a given length.

While any minor-closed property can be described by a finite set of forbidden minors, some minor-closed properties can also be described by a finite set of forbidden subgraphs, graphs that do not appear as a subgraph of any graph possessing the property. We call a graph property (which need not be minor closed) a \emph{forbidden subgraph property} (FSP) if it can be described by a finite set of forbidden subgraphs. Our main result is that the quantum query complexity of minor-closed properties depends crucially on whether the property is FSP. We show that any nontrivial minor-closed property that is not FSP has query complexity $\Theta(n^{3/2})$, whereas any minor-closed property that is FSP can be decided using $O(n^{\alpha})$ queries for some $\alpha < 3/2$, and in particular the query complexity is $o(n^{3/2})$.  In general, the value of $\alpha$ may depend on the property; we give upper bounds on $\alpha$ that depend on the sizes of certain vertex covers.

\fig{venn} summarizes our understanding of the quantum query complexity of minor-closed graph properties.  All subgraph-closed properties, which include minor-closed properties and FSPs, have an easy lower bound of of $\Omega(n)$ (\thm{subgraphlb}).  Furthermore, all sparse graph properties, which are defined in \sec{prelim} and which include all minor-closed properties, have an easy upper bound of $O(n^{3/2})$ (\thm{sparse}).  On the lower bound side, our main contribution is to show that minor-closed properties that are not FSP require $\Omega(n^{3/2})$ queries (\thm{mclb}), which tightly characterizes their quantum query complexity.  Regarding upper bounds, our main contribution is a quantum algorithm for all sparse graph properties that are FSP using $o(n^{3/2})$ queries (\cor{sparsefsp}).

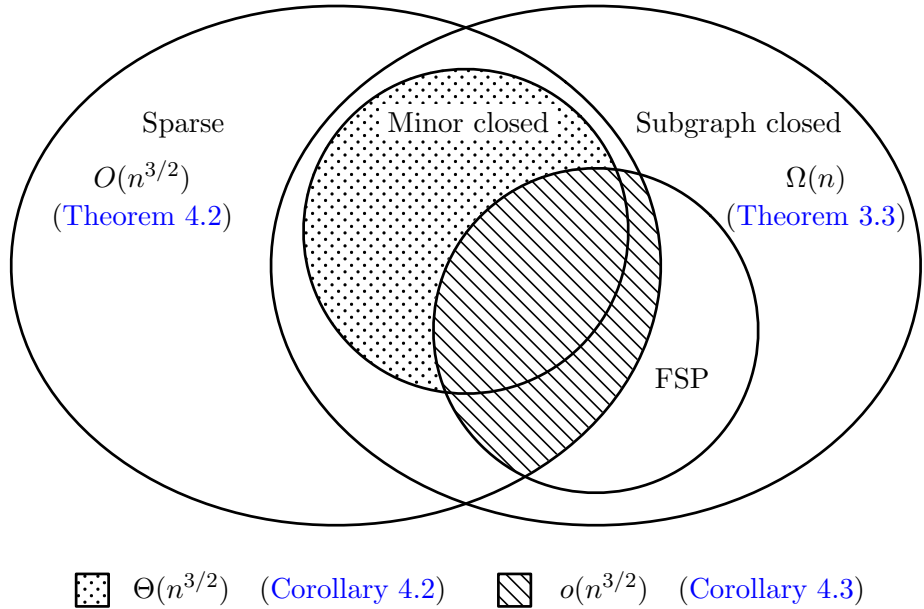
\begin{figure}
\capstart
\centering
\begin{overpic}[scale=0.68]{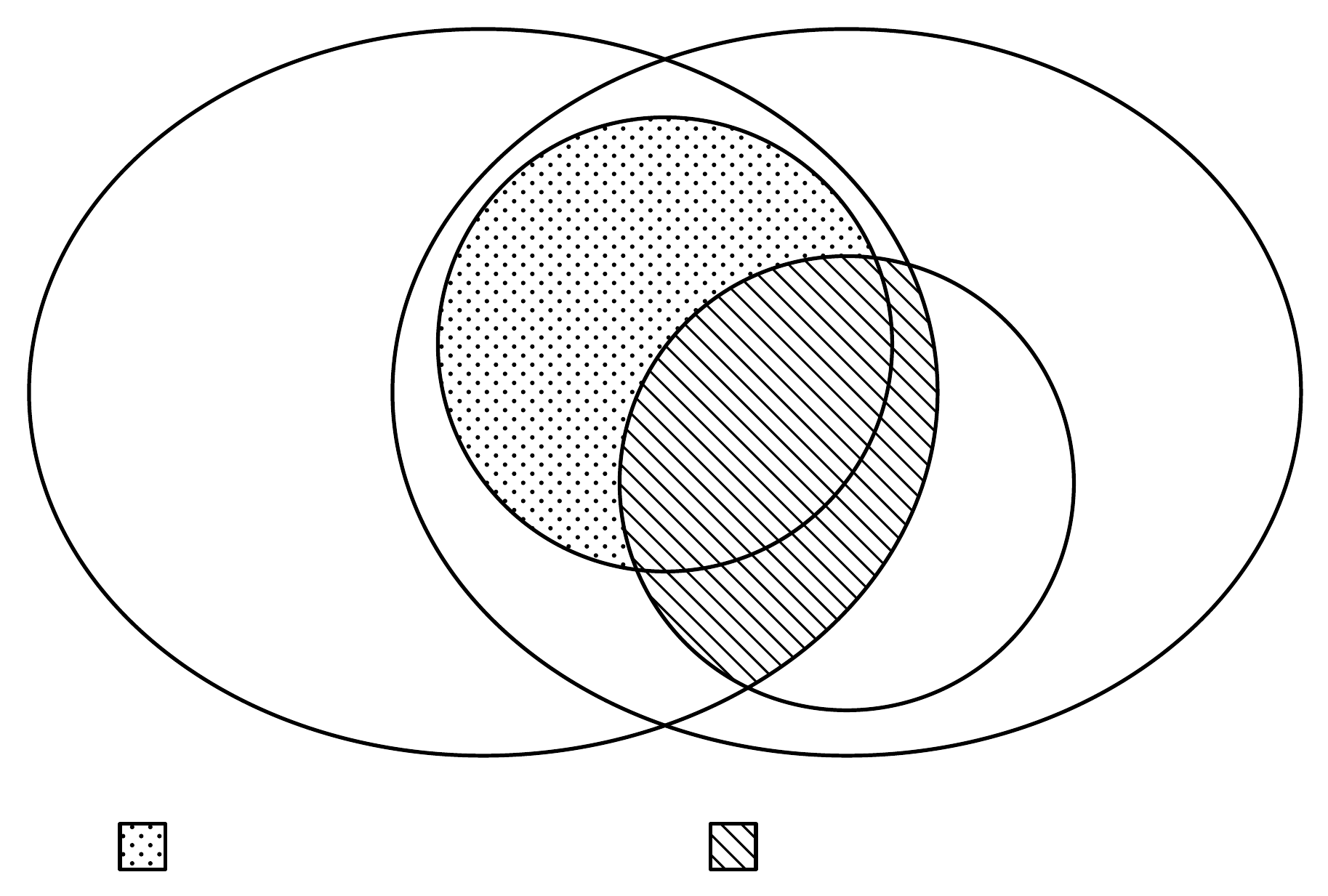} 
\put(16,52){Sparse}
\put(6,45.5){\makebox(20,4)[b]{$O(n^{3/2})$}}
\put(6,41){\makebox(20,4){(\thm{sparse})}}
\put(68,52){Subgraph closed}
\put(77,45.5){\makebox(20,4)[b]{$\Omega(n)$}}
\put(77,41){\makebox(20,4){(\thm{subgraphlb})}}
\put(41,52){\colorbox{white}{\textcolor{black}{Minor closed}}}
\put(70,25){FSP}
\put(15,2.9){$\Theta(n^{3/2})$ ~ (\cor{tightmc})}
\put(60,2.9){$o(n^{3/2})$ ~ (\cor{sparsefsp})}
\end{overpic}
\caption{Summary of the main results.
\label{fig:venn}}
\end{figure}

Our lower bounds (\sec{LB}) use the quantum adversary method~\cite{Amb02}.  The basic idea of the main lower bound is similar to the connectivity lower bound of D{\"u}rr et al.~\cite{DHHM06}.  However, it is nontrivial to show that this approach works using only the hypothesis that the property is minor closed and not FSP.  In fact, we show a slightly stronger result, assuming only that the property is not FSP and can be described by finitely many forbidden topological minors.

Our upper bounds (\sec{algo}) use the quantum walk search formalism~\cite{MNRS07}. Our approach differs from previous applications of this formalism in several respects. First, the graph underlying our quantum walk is a Hamming graph, rather than the Johnson graph used in previous algorithms.
(This difference is not essential, but it simplifies the algorithm, and a similar approach may prove useful in pedagogical treatments of other quantum walk search algorithms.)
Second, our algorithms involve several quantum walks occurring simultaneously on different Hamming graphs; although this can be viewed as a single walk on a larger graph, the salient feature is that the walks on different graphs proceed at different speeds, i.e., in each time step a different number of steps are taken on each graph. Third, our quantum walk algorithm makes essential use of the sparsity of the input graph, and to do so the algorithm must make queries even to determine which vertices of the input graph to search over (namely, to find vertices of a given degree).

The fact that our quantum walk can exploit sparsity allows us to improve upon known algorithms for many sparse graph properties, even if they are not necessarily minor closed.  In particular, we solve the $H$-subgraph containment problem with fewer queries than previous approaches for certain subgraphs $H$.  For example, we give improved algorithms for finding paths of a given length, as well as an algorithm that outperforms the general algorithm of Magniez et al.~\cite{MSS07} whenever $H$ is a bipartite graph.

Finally, as another application, we consider the $C_4$-subgraph containment problem. This can be viewed as a natural extension of the triangle problem, which is $C_3$-subgraph containment. Surprisingly, we show that $C_4$ finding can be solved with only $\tilde O(n^{1.25})$ queries, even faster than the best known  upper bound for triangle finding, which is $O(n^{1.3})$.

\section{Preliminaries}
\label{sec:prelim}

In this paper, all graphs are simple and undirected. Thus a graph on $n$ vertices is specified by $n \choose 2$ bits. In the query complexity model, the input graph is accessed by querying a black box to learn any of these $n \choose 2$ bits. Deterministic and randomized algorithms have access to a black box taking two inputs, $u$ and $v$, and returning a bit indicating whether $(u,v)$ is an edge in the graph. To accommodate quantum algorithms, we define a quantum black box in the standard way. The quantum black box is a unitary gate that maps $\ket{u,v,b}$ to $\ket{u,v,b\oplus e}$ where $(u,v)\in V\times V$, $b$ is a bit, and $e$ is 1 if and only if $(u,v)\in E$.

Let the deterministic, randomized, and quantum query complexities of determining whether a graph possesses property $\P$ be denoted as $D(\P)$, $R(\P)$, and $Q(\P)$, respectively, where for $R$ and $Q$ we consider two-sided bounded error. Clearly, $Q(\P) \leq R(\P) \leq D(\P) \leq {n \choose 2}$. Also note that these query complexities are the same for a property $\P$ and its complement $\bar\P$, since any algorithm for $\P$ can be turned into an algorithm for $\bar\P$ by negating the output, using no additional queries.

A graph property on $n$ vertices is a property of $n$-vertex graphs that is independent of vertex labeling, i.e., isomorphic graphs are considered equivalent. For a graph $G$ on $n$ vertices and an $n$-vertex graph property $\P_n$, we write $G \in \P_n$ to mean that graph $G$ has property $\P_n$. A graph property $\P \defeq \{\P_n\}_{n=1}^{\infty}$ is a collection of $n$-vertex graph properties $\P_n$ for all $n\in N$. For example, the property ``the first vertex is isolated'' is not a graph property because it depends on the labeling, and in particular it depends on which vertex we decide to call the first one. However, the property ``contains an isolated vertex'' is a graph property.

An $n$-vertex graph property $\P_n$ is nontrivial if there exists a graph that possesses it and one that does not.  A graph property $\P = \{\P_n\}_{n=1}^{\infty}$ is nontrivial if there exists an $n_0$ such that $\P_n$ is nontrivial for all $n>n_0$. Thus a property such as ``contains a clique of size 5'' is nontrivial, although it is trivial for graphs with fewer than 5 vertices.

In this paper, $K_n$ and $C_n$ refer to the complete graph and cycle on $n$ vertices, respectively. $K_{s,t}$ is the complete bipartite graph with $s$ vertices in one part and $t$ vertices in the other. A $k$-path is a path with $k$ edges (i.e., with $k+1$ vertices). The claw graph is $K_{1,3}$. Subdivisions of the claw graph can be described by specifying the number of edges in each branch of the claw: a $\{d_1,d_2,d_3\}$-claw is the subdivision of the claw with branches of $d_1$, $d_2$, and $d_3$ edges. The claw graph itself is a $\{1,1,1\}$-claw. For a graph $G$, $V(G)$ and $E(G)$ denote the vertex and edge sets of the graph; $n \defeq |V(G)|$ and $m \defeq |E(G)|$.

A graph $H$ is said to be a subgraph of $G$, denoted $H \leq_S G$, if $H$ can be obtained from $G$ by deleting edges and isolated vertices. A graph $H$ is said to be a minor of $G$, denoted $H \leq_M G$, if $H$ can be obtained from $G$ by deleting edges, deleting isolated vertices, and contracting edges. To contract an edge $(u,v)$, we delete the vertices $u$ and $v$ (and all associated edges) and create a new vertex that is adjacent to all the original neighbors of $u$ and $v$. The name ``edge contraction'' comes from viewing this operation as shrinking the edge $(u,v)$ to a point, letting the vertices $u$ and $v$ coalesce to form a single vertex.

Another way to understand graph minors is to consider reverse operations: $H \leq_M G$ if $G$ can be obtained from $H$ by adding isolated vertices, adding edges, and performing vertex splits. In a vertex split, we delete a vertex $u$ and add two new adjacent vertices $v$ and $w$, such that each original neighbor of $u$ becomes a neighbor of either $v$ or $w$, or both. In general, this operation does not lead to a unique graph, since there may be many different ways to split a vertex. 

A related operation, which is a special case of a vertex split, is known as an elementary subdivision. This operation replaces an edge $(u,v)$ with two edges $(u,w)$ and $(w,v)$, where $w$ is a new vertex.  
A graph $H$ is said to be a topological minor of $G$, denoted $H \leq_T G$, if $G$ can be obtained from $H$ by adding edges, adding isolated vertices, and performing elementary subdivisions. We call $G$ a subdivision of $H$ if it is obtained from $H$ by performing any number of elementary subdivisions.

Some graph properties can be expressed using a forbidden graph characterization.  Such a characterization says that graphs have the property if and only if they do not contain any of some set of forbidden graphs according to some notion of graph inclusion, such as subgraphs or minors. For example, a graph is a forest if and only if it contains no cycle as a subgraph, so forests are characterized by the forbidden subgraph set $\{C_k\colon k\geq 3,\, k\in\mathbb{N}\}$.  The property of being a forest can also be characterized by the single forbidden minor $C_3$, since a graph is a forest if and only if it does not contain $C_3$ as a minor. If a property can be expressed using a finite number of forbidden subgraphs, we call it a forbidden subgraph property (FSP). A property is said to be subgraph closed if every subgraph of a graph possessing the property also possesses the property. Similarly, a property is said to be minor closed if all minors of a graph possessing the property also possess the property. In a series of 20 papers spanning over 20 years, Robertson and Seymour proved the following theorem~\cite{RS04}:

\begin{theorem}[Graph minor theorem]
\label{thm:graphminor}
Every minor-closed graph property can be described by a finite set of forbidden minors.
\end{theorem}

We also require the following consequence of the graph minor theorem, which follows using well-known facts about topological minors~\cite[Theorem 2.1]{RS90}.

\begin{corollary}
\label{cor:topminor}
Every minor-closed graph property can be described by a finite set of forbidden topological minors.
\end{corollary}

We call a graph property \emph{sparse} if there exists a constant $c$ such that every graph $G$ with the property has $|E(G)| \leq c\,|V(G)|$. Nontrivial minor-closed properties are sparse, which is an easy corollary of Mader's theorem~\cite{Mad67}.

\begin{theorem}
\label{thm:mcsparse}
Every nontrivial minor-closed graph property is sparse.
\end{theorem}
\comment{
\begin{proof}
This theorem is a simple corollary of Mader's theorem~\cite{Mad67}, which states that any graph $G$ that does not contain $K_h$ as a minor must have  $|E(G)| \leq f(h) \, |V(G)|$. 

Since the family is nontrivial, there is some complete graph $K_h$ that is not contained in this family. Since the family is minor closed, no graph which contains $K_h$ as a minor is in this family either. Thus this set of graphs is a subset of the family of graphs which do not contain $K_h$ as a minor, and therefore we have $|E(G)| \leq f(h) |V(G)|$ for all graphs $G$ in this nontrivial minor-closed family.
\end{proof}
}

We use $\tilde O$ notation to denote asymptotic upper bounds that neglect logarithmic factors. Specifically, $f(n)=\tilde O(g(n))$ means $f(n)=O(g(n) \log^k g(n))$ for some constant $k$.

\section{Lower bounds}
\label{sec:LB}

In this section, we prove lower bounds on the quantum query complexity of graph properties.  We first show a simple $\Omega(n)$ lower bound for all subgraph-closed properties.  With the exception of \thm{sparse}, this covers all the properties considered in this paper, since every property considered (or its complement) is subgraph closed.  We then describe an $\Omega(n^{3/2})$ lower bound for the problem of determining whether a graph contains a cycle and explain how a similar strategy can be applied to any graph property that is closed under topological minors, provided we can identify a suitable edge of a forbidden topological minor.  Next we introduce a tool used in our general lower bounds, namely a graph invariant that is monotone with respect to topological minors.  With this tool in hand, we show an $\Omega(n^{3/2})$ lower bound for any $H$-topological minor containment problem that does not coincide with $H$-subgraph containment.  We conclude by showing the main result of this section, that every nontrivial minor-closed property $\P$ that is not FSP has $Q(\P)=\Omega(n^{3/2})$. 

The quantum lower bounds in this paper use the quantum adversary method of Ambainis~\cite{Amb02}.

\begin{theorem}[Ambainis] 
\label{thm:Ambainis}
Let $f(x_1, \ldots, x_n)$ be a function of $n$ bits and let $X, Y$ be two sets of inputs such that $f(x)\neq f(y)$ if $x\in X$ and $y\in Y$.
Let $R\subseteq X \times Y$ be a relation. Let the values $m$, $m'$, $l^{}_{x,i}$, $l'_{y,i}$ for $x \in X$, $y \in Y$ and $i \in \{1, \ldots, n\}$ be such that the following hold.
\begin{enumerate}
\item
For every $x\in X$, there are at least $m$ different $y\in Y$ such that $(x, y)\in R$.
\item
For every $y\in Y$, there are at least $m'$ different $x\in X$ such that $(x,y)\in R$.
\item
For every $x\in X$ and $i\in\{1, \ldots, n\}$, there are at most $l^{}_{x,i}$ different $y\in Y$ such that $(x, y)\in R$ and $x_i\neq y_i$.
\item
For every $y\in Y$ and $i\in\{1, \ldots, n\}$, there are at most $l'_{y,i}$ different $x\in X$ such that $(x, y)\in R$ and $x_i\neq y_i$.
\end{enumerate}
Let $l_{\max}$ be the maximum of $l^{}_{x, i}\,l'_{y, i}$ over all $(x, y)\in R$ and $i\in\{1, \ldots, n\}$ such that $x_i\neq y_i$.
Then any quantum algorithm computing $f$ with probability at least 2/3 requires $\Omega\(\sqrt{\frac{m m'}{l_{\max}}}\)$ queries. 
\end{theorem}  

For the classical lower bound in the next section, we use the following theorem of Aaronson~\cite{Aar06}.

\begin{theorem}[Aaronson] 
\label{thm:Aaronson}
Let $f,X, Y,R,m,m',l^{}_{x, i},l'_{y, i}$ be as in \thm{Ambainis}. Let $v$ be the maximum of $\min\{l^{}_{x, i}/m,l'_{y, i}/m'\}$ over all $(x, y)\in R$ and $i\in\{1, \ldots, n\}$ such that $x_i\neq y_i$. Then any randomized algorithm computing $f$ with probability at least 2/3 requires $\Omega(1/v)$ queries.
\end{theorem}

\subsection{Subgraph-closed properties}

We begin with the $\Omega(n)$ lower bound for all nontrivial subgraph-closed graph properties. 

\begin{theorem}
\label{thm:subgraphlb}
For any nontrivial subgraph-closed graph property $\P$, $Q(\P)=\Omega(n)$, $R(\P)=\Theta(n^2)$, and $D(\P)=\Theta(n^2)$.
\end{theorem}
\begin{proof}
Since $\P$ is nontrivial, for all $n\geq n_0$ there exists a graph on $n$ vertices that is in $\P$, and since $\P$ is subgraph closed, this implies that the empty graph on $n$ vertices is in $\P$. Since $\P$ is nontrivial, there exists a graph $H$ not in $\P$.  Since $H \notin \P$, all supergraphs of $H$ are also not in $\P$. Now let $d = \max\{|V(H)|, n_0\}$. Then the empty graph on $d$ vertices is in $\P$, while the complete graph on $d$ vertices, $K_d$, is not in $\P$.

We prove a lower bound for the problem of distinguishing the empty graph on $n$ vertices and the $n$-vertex graph formed by $K_d$ and $n-d$ isolated vertices, for all $n \geq d$.

Let $X$ contain only one graph, the empty graph on $n$ vertices. Let $Y$ contain all graphs on $n$ vertices that contain exactly one copy of $K_d$ and $n-d$ isolated vertices. Let $R$ be the trivial relation, $R=X \times Y$. Thus, in the notation of \thm{Ambainis}, $m$ equals  $|Y|={n \choose d}$. Similarly, $m' = |X| = 1$. Since $X$ contains only the empty graph, $l^{}_{x,i}$ counts the number of graphs in $Y$ in which the \th{i} edge is present. Since one edge fixes two vertices of $K_d$, the number of related graphs in $Y$ with a given edge is $n-2 \choose d-2$. Since $X$ contains only 1 graph, $l'_{y,i}\leq 1$.

The quantum adversary method (\thm{Ambainis}) gives us $Q(\P)=\Omega\(\sqrt{{{n\choose d}}/{{n-2 \choose d-2}}}\)=\Omega(n)$. \thm{Aaronson} gives us $R(\P)=\Omega(n^2)$ and thus $D(\P)=\Omega(n^2)$. Since all query complexities are at most $n \choose 2$, we have $D(\P)=\Theta(n^2)$ and $R(\P)=\Theta(n^2)$.
\end{proof}

Note that this theorem can also be proved by reduction from the unstructured search problem using Tur{\'a}n's theorem. By Tur{\'a}n's theorem, the densest graph on $n$ vertices that does not contain $K_d$ as a subgraph has $\Theta(n^2)$ non-edges. Thus any algorithm that decides the graph property must be able to distinguish the densest graph from the same graph with one extra edge. Since there are $\Theta(n^2)$ non-edges, the problem of searching a space of size $\Theta(n^2)$ can be reduced to this, giving the appropriate lower bounds.

This theorem shows that all the properties considered in this paper are classically uninteresting from the viewpoint of query complexity, since their classical (deterministic or randomized) query complexity is $\Theta(n^2)$. Since the classical query complexity of all these properties is known, from now on we focus on quantum query complexity.

\subsection{Acyclicity}

We now show an $\Omega(n^{3/2})$ lower bound for the problem of determining whether a graph is acyclic (i.e., a forest).  We then isolate the main idea of this proof, formulating a lemma that is useful for establishing lower bounds for more general topological minor-closed graph properties.

The lower bound for acyclicity is similar to the connectivity lower bound of D{\"u}rr et al.~\cite{DHHM06}. The intuition is that a long path and a long cycle look the same locally. Since algorithms only have access to local information, these two graphs should be hard to distinguish. Unfortunately this is not sufficient, since a path can be easily distinguished from a cycle by searching the entire graph for a degree-1 vertex, which can be done with $O(n)$ queries. Instead, we try to distinguish a  path from the disjoint union of a cycle and a path. Now both graphs have 2 degree-1 vertices. We require both the cycle and the path to be long, since a short cycle or path could be quickly traversed. 

\begin{theorem}
\label{thm:C3minor}
Deciding if a graph is a forest requires $\Omega(n^{3/2})$ queries.
\end{theorem}
\begin{proof}
Let $X$ be the set of all graphs on $n$ vertices isomorphic to the path on $n$ vertices. Let $Y$ be the set of all graphs on $n$ vertices that are the disjoint union of a path and a cycle, such that no vertex is isolated and both the cycle and path have more than $n/3$ vertices. Clearly all graphs in $X$ are forests and all graphs in $Y$ are not. Let $(x,y)\in R$ if there exist 4 vertices $a,b,c,d$, such that the only difference between $x$ and $y$ is that in $x$, the induced subgraph on these vertices has only the edges $(a,b)$ and $(c,d)$, but in $y$, the induced subgraph has only edges $(a,c)$ and $(b,d)$. See \fig{pathcycle} for an example of two related graphs on $n=10$ vertices.

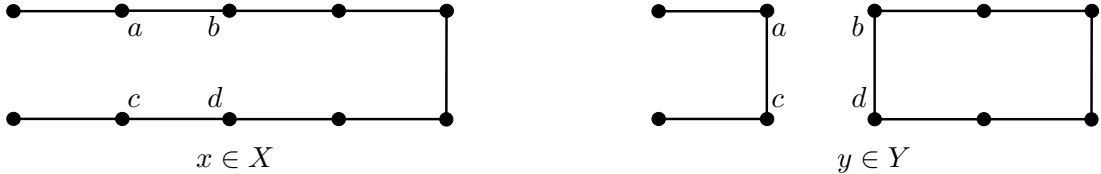
\begin{figure}
\capstart
\centering
\begin{overpic}[scale=0.8]{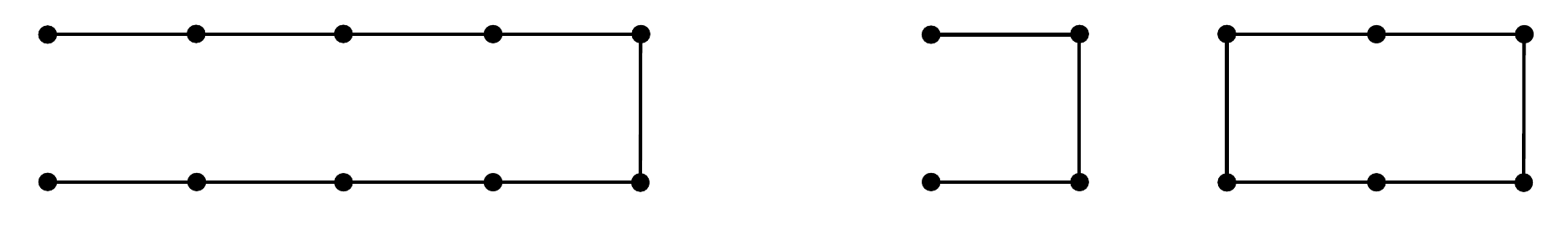} 
\put(13,11.5){$a$}
\put(20,11.5){$b$}
\put(13,5.25){$c$}
\put(20,5.25){$d$}
\put(69.25,11.5){$a$}
\put(76.25,11.5){$b$}
\put(69.25,5.25){$c$}
\put(76.25,5.25){$d$}
\put(19,0){$x \in X$}
\put(75,0){$y \in Y$}
\end{overpic}
\caption{An example of two graphs on 10 vertices such that $(x,y) \in R$.
\label{fig:pathcycle}}
\end{figure}

Now let us compute the relevant quantities from \thm{Ambainis}. Recall that $m$ is the minimum number of graphs $y \in Y$ that each graph $x \in X$ is related to. Each graph in $X$ can be transformed to a related graph in $Y$ by selecting edges $(a,b)$ and $(c,d)$ such that the distance between the edges is between $n/3$ and $2n/3$. There are $n-1$ ways to pick the edge $(a,b)$, and for any choice of $(a,b)$ there are $n/3$ possible edges $(c,d)$, which gives $m=\Omega(n^2)$. Similarly, $m'=\Omega(n^2)$, since in a graph $y\in Y$, we can choose any one edge in the cycle to be $(a,c)$ and any one in the path to be $(b,d)$.

Now let us compute $l_{\max}$. The quantity $l^{}_{x,i}$ counts the number of graphs in $Y$ that are related to $x$ and differ from $x$ at the \th{i} edge. The variable $i$ indexes the bits of the adjacency matrix; let $x_i$ indicate whether the \th{i} edge is present or absent in $x$. Since $l_{\max}$ is the maximum of $l^{}_{x,i}\,l'_{y,i}$ over all related pairs $(x,y)$ such that $x_i \neq y_i$, let us first compute the maximum of $l^{}_{x,i}\,l'_{y,i}$ over all related pairs $(x,y)$ where $x_i = 0$ and $y_i = 1$.

Since $x_i = 0$, the \th{i} edge is not present in $x$ and is present in the related graphs $y$, so without loss of generality the \th{i} edge is $(a,c)$. To obtain a graph in $Y$, we need to choose vertices $b$ and $d$ such that $(a,b)$ and $(c,d)$ are edges in $x$. We can choose either of $a$'s two neighbors to be $b$ and either of $c$'s two neighbors to be $d$. This gives at most 4 different $y \in Y$ that are related to this $x$ and differ at the \th{i} edge. Thus $l^{}_{x,i} \leq 4$ when $x_i = 0$. (The reason for the inequality is that sometimes this may not yield a valid graph in $y$, e.g., when the resulting graph is not of the appropriate form, or when the resulting path or cycle is too short.)

Since $y_i=1$, $l'_{y,i}$ counts the number of graphs in $X$ related to $y$ that do not have the \th{i} edge. Again, we can assume this edge is $(a,c)$, since it is present in $y$ but not in related graphs in $X$.  If the \th{i} edge is an edge on the path in $y$ (as opposed to the cycle), then choosing any edge in the cycle will give two vertices $b$ and $d$ that give rise to valid graphs in $x$ when the edges $(a,b)$ and $(c,d)$ are added and $(a,c)$ and $(b,d)$ are deleted. Thus there are $O(n)$ possible choices for $b$ and $d$, which gives $O(n)$ related graphs in $X$. If $(a,c)$ is an edge on the cycle, then we can choose any edge on the path as $(b,d)$. This again leads to at most $O(n)$ possibilities, which gives us $l'_{y,i} = O(n)$. Thus $l^{}_{x,i}\,l'_{y,i} = O(n)$ for all $(x,y) \in R$ when $x_i = 0$ and $y_i = 1$.

Now we need to compute $l^{}_{x,i}\,l'_{y,i}$ when $x_i = 1$ and $y_i = 0$. The values of  $l^{}_{x,i}$ and $l'_{y,i}$ are similar for this case: one is $O(1)$ and the other is $O(n)$. This gives $l^{}_{x,i}\,l'_{y,i}=O(n)$ for all $(x,y) \in R$ such that $x_i \neq y_i$. Thus $l_{\max} = O(n)$.

Using \thm{Ambainis}, we get a lower bound of $\Omega(\sqrt{{n^4}/{n}}) = \Omega(n^{3/2})$. 
\end{proof}

Cyclicity can be viewed as the property of containing $C_3$ as a minor, or equivalently, as a topological minor.  Note that for any constant $r$, the same proof also works for the property of containing $C_r$ as a minor (i.e., containing a cycle of length at least $r$), since the graphs in $X$ did not contain any cycle, and the graphs in $Y$ contained a cycle of length at least $n/3$, which contains any constant-sized cycle as a minor.

Inspecting the proof technique more closely, we see that no particular property of forests was used, other than the facts that all subdivisions of $C_3$ are not forests, and that if we delete an edge from a subdivision, the resulting graph is a forest. More precisely, we used the existence of a graph $G$ (in this case $C_3$) and an edge $(u,v) \in E(G)$ (since $C_3$ is edge transitive it does not matter which edge is chosen) such that if $(u,v)$ is subdivided any number of times, the resulting graph still does not have the property (in this case, of being a forest) and if $(u,v)$ is replaced by two disjoint paths the resulting graph does have the property.  The following lemma formalizes this intuition.

\begin{lemma}
\label{lem:mainlb}
Let $\P$ be a graph property closed under topological minors. If there exists a graph $G \notin \P$ and an edge $(u,v)$ in $G$, such that replacing the edge $(u,v)$ by two disjoint paths of any length, one connected to vertex $u$ and the other connected to vertex $v$, always results in a graph $G_1 \in \P$, then $Q(\P) = \Omega(n^{3/2})$.
\end{lemma}
\begin{proof}
The proof is similar in structure to \thm{C3minor} and subsumes it if we take $G=C_3$. For the general case, let $G$ be a graph on $k$ vertices.

Let $G_1$ be the graph $G$ with the edge $(u,v)$ deleted, a path of length $n_1$ attached at vertex $u$, and a path of length $n_2$ attached at vertex $v$, such that $n_1, n_2 \geq n/3$ and $|V(G_1)| = n$.  By assumption, $G_1 \in \P$.  Let $X$ be the set of all graphs isomorphic to $G_1$.

Let $G_2'$ be the graph $G$ with the edge $(u,v)$ subdivided $n_1$  times, where $n_1 \geq n/3$. This is equivalent to replacing the edge by a path of length $n_1+1$. Let $G_2$ be the disjoint union of $G_2'$ and a path of length $n_2$, such that $n_2 \geq n/3$ and $|V(G_2)| = n$. Clearly $G_2 \notin \P$, since it contains $G$ as a topological minor.  Let $Y$ be the set of all graphs isomorphic to $G_2$. 

As before, let $(x,y)\in R$ if there exist 4 vertices $a,b,c,d$, such that the only difference between $x$ and $y$ is that in $x$, the induced subgraph on these vertices has only the edges $(a,b)$ and $(c,d)$, but in $y$, the induced subgraph has only the edges $(a,c)$ and $(b,d)$. 

Each graph in $X$ can be transformed to a related graph in $Y$ by first selecting an edge $(a,b)$ in the first path and then picking another edge $(c,d)$ in the second path. Each path contains more than $n/3$ edges, but we have to satisfy the condition that when the edges $(a,b)$ and $(c,d)$ are removed and replaced with $(a,c)$ and $(b,d)$, the resulting graph is in $Y$. This means that after swapping the edges, both the long disjoint path and the path between vertices $u$ and $v$ have to be longer than $n/3$. Even with this restriction there are $\Omega(n^2)$ graphs in $Y$ related to any graph $x \in X$. For example, we can choose any edge on the shorter path to be $(a,b)$, and then there are at least $n/6$ edges on the longer path which can be chosen to be $(c,d)$, which will give a graph in $Y$ when the edges $(a,b)$ and $(c,d)$ are removed and replaced with $(a,c)$ and $(b,d)$.

Similarly, $m'=\Omega(n^2)$. We have to choose an edge from the disjoint path and one from the path which connects the vertices $u$ and $v$. Again, we can choose any edge from the smaller of the two paths, and then there are still at least $n/6$ edges in the other path left to choose, such that if those edges are chosen as  $(a,c)$ and $(b,d)$, and then we perform the swap (i.e.,  $(a,c)$ and $(b,d)$ are removed and replaced with  $(a,b)$ and $(c,d)$), this results in a graph in $X$. 

Now let us upper bound the maximum of $l^{}_{x,i}\,l'_{y,i}$ over all related pairs $(x,y)$ where $x_i = 0$ and $y_i = 1$. Since the \th{i} edge is not present in $x$, we can take this to be the edge $(a,c)$. To obtain a graph in $Y$, we need to choose vertices $b$ and $d$ such that $(a,b)$ and $(c,d)$ are edges in $x$. We can choose any one of $a$'s neighbors to be $b$ and any one of $c$'s neighbors to be $d$. (For some edges, this procedure may not give a graph in $Y$, but we only need an upper bound.) Since $a$ and $c$ have $O(1)$ neighbors, $l^{}_{x,i} = O(1)$ when $x_i = 0$. 

To compute $l'_{y,i}$, we can assume the \th{i} edge is the edge $(a,c)$, since it is present in the graph $y$ but not in related graphs in $X$.  If the  \th{i} edge is an edge on the disjoint path in $y$ or the path connecting vertices $u$ and $v$, there can be at most $O(n)$ choices for the edge $(b,d)$ on the other path that gives rise to a valid graph in $X$ when the edges $(a,b)$ and $(c,d)$ are added and $(a,c)$ and $(b,d)$ are deleted. As before, sometimes there may be no related graphs with the \th{i} edge absent, but we only require an upper bound. Thus $l^{}_{x,i}\,l'_{y,i} = O(n)$ for all $(x,y) \in R$ when $x_i = 0$ and $y_i = 1$.

Now we need to compute $l^{}_{x,i}\,l'_{y,i}$ when $x_i = 1$ and $y_i = 0$. The values of  $l^{}_{x,i}$ and $l'_{y,i}$ are similar for this case: one is $O(1)$, and the other is $O(n)$. This gives $l^{}_{x,i}\,l'_{y,i}=O(n)$ for all $(x,y) \in R$ such that $x_i \neq y_i$. Thus $l_{\max} = O(n)$. 

Using \thm{Ambainis}, we get a lower bound of $\Omega(\sqrt{{n^4}/{n}}) = \Omega(n^{3/2})$.
\end{proof}

\subsection{A graph invariant for topological minor containment}

To identify suitable edges for use in \lem{mainlb}, we introduce a graph invariant that is monotone with respect to topological minors.  As a simple application, we use this invariant to show an $\Omega(n^{3/2})$ lower bound for all $H$-topological minor containment properties that are not also $H$-subgraph containment properties.

Call an edge \emph{internal} if it lies on a cycle or on a path joining two vertices of degree 3 or greater. Call all other edges \emph{external}. Note that an edge is external if and only if it belongs to a \emph{dangling path}, a vertex subset $\{v_1,v_2,\ldots,v_k\}$ such that $v_1$ is adjacent to $v_2$, $v_2$ is adjacent to $v_3$, etc., where $v_1$ has degree $1$ and $v_2,v_3,\ldots,v_{k-1}$ have degree $2$.  For a graph $G$, let $\beta(G)$ denote the number of internal edges in $G$.  We claim that $\beta$ is monotone with respect to topological minors.

\begin{lemma}
If $H \leq_T G$ then $\beta(H) \leq \beta(G)$.
\end{lemma}

\begin{proof}
Let $H$ be a topological minor of $G$, where $G$ is obtained from $H$ in one step, i.e., by performing an elementary subdivision or by adding a single edge or a single isolated vertex. We show that $\beta(H) \leq \beta(G)$.  Then the same inequality follows when $H$ is obtained from $G$ by any number of steps.

It is clear that adding an isolated vertex does not change the $\beta$ value of a graph. Adding an edge to $H$ results in a supergraph of $H$, which contains all the high degree vertices and cycles that $H$ does (and possibly more). Therefore each internal edge in $H$ remains an internal edge in $G$, which shows that the $\beta$ value cannot decrease.

Finally, since subdividing an edge in $H$ replaces the edge with a path of length 2, this does not change the connectivity of the graph. Any paths that used the subdivided edge can now use the path of length 2 instead. All the vertices of $H$ have the same degree in $G$. Thus cycles cannot be destroyed by subdivision, and neither can a path between any two particular vertices.
\end{proof}

We use a specific topological minor operation that strictly decreases the invariant.

\begin{lemma}
\label{lem:decreasebeta}
In a graph $G$, deleting an internal edge $(u,v)$ and adding two disjoint paths, one starting from vertex $u$ and one from $v$, decreases the $\beta$ value of the resulting graph $H$.
\end{lemma}

\begin{proof}
It is clear that merely deleting the internal edge $(u,v)$ decreases the $\beta$ value of the graph. Let us now show that every internal edge in $H$ is also an internal edge in $G$.

First, observe that none of the edges that were added to $H$ as part of the two disjoint paths are internal. This follows because the added edges lie on a path with one end connected to $u$ or $v$ and the other end free. No edge on this path is part of a cycle, and the path does not contain any vertices of degree 3 or more.

It remains to consider edges of $H$ that are also present in $G$.  If an edge is internal in $H$ because it lies on a cycle, then it is also internal in $G$ since all cycles in $H$ are present in $G$. If an edge is internal in $H$ because it lies on a path between two vertices of degree 3 or more, none of the vertices on that path can be on the added disjoint paths, since all the added vertices have degree 1 or 2, and all vertices that are present in both $H$ and $G$ have exactly the same degree in both graphs. Thus such an edge is internal in $G$ as well.  Since $G$ has all the internal edges of $H$, and at least one more (the edge $(u,v)$), $\beta(H) < \beta(G)$.
\end{proof}

Finally, the graphs with $\beta(H)=0$ have a nice characterization.

\begin{lemma}
\label{lem:betazero}
For any graph $H$, $H$-topological minor containment is equivalent to $H$-subgraph containment if and only if $\beta(H)=0$.
\end{lemma}

\begin{proof}
If $H$ is a subgraph of another graph, then it is also a topological minor of that graph. Thus to show the equivalence of topological minor containment and subgraph containment, we only have to show that if $H$ is a topological minor of $G$ then $H$ is also a subgraph of $G$.

If $\beta(H)=0$, then $H$ contains no internal edges, which means each connected component of $H$ is a subdivision of a star graph.
Then it is easy to see that every subdivision of $H$ contains $H$ as a subgraph.

To show the converse, note that if $H$ is a graph in which some connected component is not a subdivision of a star, then $H$ must have 2 vertices of degree 3 in a connected component or have a cycle. For these graphs we exhibit a subdivision of $H$ which does not contain $H$ as a subgraph. 

Let $H$ be a graph on $k$ vertices that contains a cycle. Let $G$ be the graph obtained by performing $k$ subdivisions on every edge of $H$. Now the smallest cycle in $G$ is at least $k$ times longer than the smallest cycle in $H$. Thus $G$ contains no cycles of length less than or equal to $k$. But $H$ is a graph on $k$ vertices, and contains a cycle, which must be of length less than or equal to $k$. Therefore $H$ cannot be a subgraph of $G$.

Finally, let $H$ be a $k$-vertex acyclic graph with 2 or more vertices of degree 3 in the same connected component. Again let $G$ be the graph obtained by performing $k$ subdivisions on every edge of $H$. Now the shortest path joining any pair of degree-3 vertices in $G$ has length greater than $k$. However, any path joining 2 degree-3 vertices in $H$ has length less than or equal to $k$. Thus $H$ cannot be a subgraph of $G$, since deleting edges and isolated vertices cannot decrease the shortest path between two vertices. 
\end{proof}

Using this invariant together with \lem{mainlb}, we can easily show an $\Omega(n^{3/2})$ lower bound for $H$-topological minor containment assuming that this property does not coincide with $H$-subgraph containment. Since \lem{betazero} characterizes such graphs, we know that $H$ must be cyclic or contain 2 vertices of degree at least 3 in the same connected component.

\begin{theorem}
For all graphs $H$, if $H$-topological minor containment is not equivalent to $H$-subgraph containment, then the quantum query complexity of $H$-topological minor containment is $\Omega(n^{3/2})$.
\end{theorem}

\begin{proof}
To apply \lem{mainlb}, $\P$ must be closed under topological minors.  Thus, we consider the property of \emph{not} containing $H$ as a topological minor.

We require a graph $G\notin \P$ with the properties stated in \lem{mainlb}. Let $G$ be the graph $H$ itself, and let $(u,v)$ be any internal edge, which must exist by \lem{betazero}.  By \lem{decreasebeta}, replacing $(u,v)$ with two disjoint paths results in a graph $G'$ with $\beta(G')<\beta(H)$. Thus $G'$ cannot contain $H$ as a minor, which gives $G' \in \P$, and \lem{mainlb} gives us the $\Omega(n^{3/2})$ lower bound.
\end{proof}

\subsection{Minor-closed properties}

We are now ready to prove our main lower bound result, an $\Omega(n^{3/2})$ lower bound for any minor-closed graph property that is not FSP.  By \cor{topminor}, any minor-closed graph property can be described in terms of forbidden topological minors.  However, so far, we have only considered the case of a single forbidden topological minor.  With multiple forbidden topological minors, some internal edges of some forbidden minors may not suffice for use in \lem{mainlb}, since subdividing an internal edge of one minor may result in a graph that contains one of the other minors.  Hence our main challenge is to identify a suitable edge for use in \lem{mainlb}.

We now introduce some terminology that will only be used in this subsection. We call a set of graphs $B$ \emph{equivalent under topological minor containment and subgraph containment} if whenever a graph in $B$ is a topological minor of a graph $G$, there is some graph in $B$ that is also a subgraph of $G$. 

\begin{lemma}
For any graph property $\P$ that is not FSP and that is described by a finite set of forbidden topological minors, there exists a graph $G \notin \P$ and an edge $(u,v) \in E(G)$ satisfying the conditions of \lem{mainlb}.
\end{lemma}

\begin{proof}
We define $\P$ using two finite sets of forbidden topological minors, $S$ and $B$, where the set $B$ is equivalent under topological minor containment and subgraph containment. Clearly, such a description exists, because we can take $B$ to be the empty set and let $S$ be the set of forbidden topological minors.

Among all possible descriptions of the property $\P$ in terms of finite forbidden sets $S$ and $B$, we choose a description that minimizes $|S|$. Since $\P$ is not FSP, it cannot be described by a pair $S$ and $B$ where $S = \emptyset$ since $B$ would then provide a forbidden subgraph characterization of $\P$. Let $l \defeq |S| \neq 0$.  Order the graphs in $S$ by their $\beta$ values, so that $S=\{H_1, H_2, \ldots H_l\}$ where $\beta(H_1) \leq \beta(H_2) \leq \ldots \leq \beta(H_l)$. 

We claim that $H_1$ can serve as the required graph $G$ for \lem{mainlb}. Clearly, $H_1 \notin \P$. $H_1$ must contain an internal edge, since otherwise \lem{betazero} implies that $H_1$ is equivalent under subgraph and topological minor containment, in which case $H_1$ could be removed from $S$ and added to $B$, violating the minimality of $S$. It remains to show that one of the internal edges of $H_1$ satisfies the conditions of \lem{mainlb}.

Toward a contradiction, assume that none of the internal edges of $H_1$ could serve as the edge $(u,v)$. This means that for each internal edge there is a pair of disjoint paths such that the graph resulting from replacing the edge with this pair of paths, $G'$, does not possess property $\P$. Since $G' \notin \P$, $G'$ must contain some graph in $S$ or $B$ as a topological minor. Since an internal edge was deleted and replaced with two disjoint paths, the $\beta$ value of the resulting graph has decreased (by \lem{decreasebeta}). Since $\beta(G') < \beta(H_1)$ and $\beta(H_1) \leq \beta(H_i)$ for all $1 \leq i \leq l$, none of the graphs in $S$ can be a topological minor of $G'$, and thus only a graph in $B$ can be a topological minor of $G'$.

Hence, for each edge $(u,v)$, there exists a pair of disjoint paths such that when $(u,v)$ is replaced by these paths, the resulting graph $G'$ contains one of the graphs in $B$ as a topological minor, and therefore as a subgraph. Let $G''$ be a supergraph of $G'$ obtained by adding an extra edge that connects the loose ends of the two disjoint paths. Since $G''$ is obtained by replacing the edge $(u,v)$ by a long path, it is a subdivision of $H_1$. Since $G'$ contains a graph in $B$ as a subgraph, so does $G''$.

It follows that for every internal edge $(u,v)$ of $H_1$, there is a constant $a$ such that if the edge $(u,v)$ is subdivided $a$ or more times, the resulting graph contains some graph from $B$ as a subgraph. Let the maximum constant $a$ over all internal edges be $c$. If any internal edge of $H_1$ is subdivided more than $c$ times, it contains some graph from $B$ as a topological minor. We use this fact to get a forbidden subgraph characterization for $\{H_1\} \cup B$.

Let the number of internal edges in $H_1$ be $k$.  Let $J$ be any graph that contains some graph from $\{H_1\} \cup B$ as a topological minor. If it contains some graph from $B$ as a topological minor, it also contains some graph from $B$ as a subgraph, so $B$ already has a forbidden subgraph characterization. The only graph whose containment as a topological minor we have to express by a forbidden subgraph characterization is $H_1$. So let $J$ contain $H_1$ as a topological minor. Since some subdivision of $H_1$ is a subgraph of $J$, let $J'$ be a minimal subgraph of $J$ that is a subdivision of $H_1$, i.e., no subgraph of $J'$ is a subdivision of $H_1$.

We claim that if $J'$ has more than $ck+|V(H_1)|$ vertices, then it already contains some graph from $B$ as a subgraph. Since $J'$ is a subdivision of $H_1$ and has $ck$ more vertices than $H_1$, it must be obtained after $ck$ subdivisions. Moreover, no subgraph of $J'$ can be a subdivision of $H_1$. Thus $J'$ must be obtained by subdividing only internal edges of $H_1$, since subdividing an external edge leads to a supergraph of the original graph (because an external edge must be on a dangling path). Since $J'$ is obtained from $H_1$ by $ck$ subdivisions of internal edges, at least one internal edge was divided $c$ times. Let the graph with this edge divided $c$ times be called $H'$. Since the order of subdivisions does not matter, $H' \leq_T J'$. However, by assumption there is a graph in $B$ that is a topological minor of $H'$. By the transitivity of topological minor containment, there is a graph in $B$ that is a topological minor of $J'$. But since $B$-subgraph containment and $B$-topological minor containment are equivalent, there is a graph in $B$ that is a subgraph of $J'$.  Thus we do not need to forbid any additional subgraphs in order to exclude graphs $J'$ with more than $ck+|V(H_1)|$ vertices.

Now suppose that $J'$ has fewer than $ck+|V(H_1)|$ vertices. Let $F$ be the set of all subdivisions of $H_1$ with at most $ck+|V(H_1)|$ vertices. Clearly $J' \in F$, and $F$ is a finite set of graphs. The set $F \cup B$ is now a forbidden subgraph characterization for the property of not containing any graph from $\{H_1\} \cup B$ as a topological minor.

This gives us a different representation of $\P$, using the new sets $S' = \{H_2,H_3,\ldots, H_l\}$ and $B'=F \cup B$. But $|S'|<l$, which contradicts the minimality of the original characterization.
\end{proof}

Combining this lemma with \cor{topminor} and \lem{mainlb}, we get the main result of this section.

\begin{theorem}
\label{thm:mclb}
For any nontrivial minor-closed property $\P$ that is not FSP, $Q(\P)=\Omega(n^{3/2})$.
\end{theorem}

This lower bound cannot be improved due to a matching algorithm shown in \sec{algo}.  It cannot be extended to minor-closed properties that are also FSP because, as we also show in \sec{algo}, every property of this type has query complexity $o(n^{3/2})$.

Since the complement of $H$-minor containment is minor closed, we have the following.

\begin{corollary}
If $H$ is a graph for which the property of not containing $H$ as a minor is not FSP, then the quantum query complexity of $H$-minor containment is $\Omega(n^{3/2})$.
\end{corollary}

Note that $H$-minor containment is not FSP for most graphs $H$. The only exceptions are graphs in which each connected component is a path or a subdivision of a claw ($K_{1,3}$). It is not hard to show that if $H$ is such a graph, then $H$-minor containment is equivalent to $H$-subgraph containment. For such graphs $H$, one can show that $H$ is a minor of $G$ if and only if it is a topological minor of $G$~\cite[Proposition 1.7.4]{Die05}. Then, by \lem{betazero}, such an $H$ is a topological minor of $G$ if and only if it is a subgraph, which proves the claim. In \sec{algo} we specifically study the query complexity of $H$-subgraph containment for such graphs (see \cor{pathfinding} and \thm{paths7-10}), as well as for general minor-closed properties that are FSP.

\section{Algorithms}
\label{sec:algo}

We now turn to quantum algorithms for deciding minor-closed graph properties, as well as related algorithms for subgraph-finding problems.

\subsection{Sparse graph detection and extraction}

We begin by describing some basic tools that allow us to detect whether a graph is sparse and to optimally extract the adjacency matrix of a sparse graph.

To tell whether a graph is sparse, we can apply quantum counting to determine approximately how many edges it contains.  In particular, Theorem 15 of \cite{BHMT00} gives the following.
\begin{theorem}[Approximate quantum counting]\label{thm:qcounting}
Let $f\colon \{1,\ldots,N\} \to \{0,1\}$ be a black-box function with $|f^{-1}(1)|=K>0$, and let $\epsilon \in (0,1]$.  There is a quantum algorithm that produces an estimate $\tilde K$ of $K$ satisfying $|K-\tilde K| \le \epsilon K$ with probability at least $2/3$, using $O(\frac{1}{\epsilon}\sqrt{N/K})$ queries of $f$.  If $K=0$, the algorithm outputs $\tilde K=0$ with certainty in $O(\sqrt N)$ queries.
\end{theorem}

Note that the failure probability can be reduced to $\delta$ by repeating this algorithm $O(\log\frac{1}{\delta})$ times and taking the median of the resulting estimates.

Applying \thm{qcounting} to approximately count the edges of a graph, we have the following.

\begin{corollary}\label{cor:sparsedetect}
For any constant $\epsilon>0$ and function $f\colon \Z^{+} \to \Z^{+}$ there is a quantum algorithm using $O(\sqrt{n^2/f(n)} \log\frac{1}{\delta})$ queries which accepts graphs with $m \ge (1+\epsilon)f(n)$ and rejects graphs with $m \le (1-\epsilon)f(n)$ with probability at least $1-\delta$.
\end{corollary}

\begin{proof}
Approximate quantum counting with accuracy $\epsilon$ can distinguish the two cases. However, if $m  \ll f(n)$, then quantum counting requires $O(\sqrt{n^2/m})$ queries, much more than the claimed $O(\sqrt{n^2/f(n)})$ queries. This can be fixed by adding $n$ more vertices and $f(n)$ edges so the total edge count is always greater than $f(n)$. Now we have to distinguish the cases $m \ge (2+\epsilon)f(n)$ and $m \le (2-\epsilon)f(n)$. This can be done using accuracy $\epsilon/2$ and only $O(\sqrt{n^2/f(n)})$ queries.

This procedure has constant success probability. Repeating this $O(\log\frac{1}{\delta})$ times and taking the majority vote of the outcomes boosts the success probability to at least $1-\delta$.
\end{proof}

We also use a procedure for extracting all marked items in a search problem.

\begin{lemma}\label{lem:searchall}
Let $f\colon \{1,\ldots,N\} \to \{0,1\}$ be a black-box function with $|f^{-1}(1)|=K$.  The bounded-error quantum query complexity of determining $f^{-1}(1)$ is $O(\sqrt{NK})$ if $K>0$, and $O(\sqrt N)$ if $K=0$.
\end{lemma}

This result and its optimality appear to be folklore (see for example \cite{Amb10}); we include a short proof for completeness. It can also be proved using the techniques of~\cite{AIN+08} or~\cite{DHHM06}.

\begin{proof}
First check if $K=0$ by standard Grover search, using $O(\sqrt N)$ queries; if so, we are done.
Otherwise, by Theorem 17 of \cite{BHMT00}, we can exactly determine $K$ with bounded error in $O(\sqrt{NK})$ queries.  By Theorem 16 of \cite{BHMT00}, given $K$, we can find a marked item with certainty in $O(\sqrt{N/K})$ queries.  We repeat this algorithm $K$ times, unmarking each item after we find it, until there are no more marked items.  The number of queries used by this procedure is $O( \sum_{i=0}^{K-1} \sqrt{N/(K-i)})$.
Observe that
\begin{align}
  \sum_{i=0}^{K-1} \sqrt{\frac{N}{K-i}}
  \le \sqrt{N} \int_0^K \frac{\d{x}}{\sqrt x}
  = 2 \sqrt{NK}.
\end{align}
Thus $O(\sqrt{NK})$ queries suffice to find the $K$ marked items.
\end{proof}

Applying these results, we find that sparse graph properties can be decided in $O(n^{3/2})$ queries.

\begin{theorem}\label{thm:sparse}
If $\P$ is a sparse graph property, then $Q(\P) = O(n^{3/2})$.
\end{theorem}

\begin{proof}
Since $\P$ is sparse, there is a constant $c$ such that $G \in \P$ implies $m \le cn$. By \cor{sparsedetect}, we can reject graphs with $m \geq 2cn$ and keep for further consideration those with $m \leq cn$ with bounded error using $O(\sqrt{n})$ queries. (It does not matter whether graphs with $cn < m < 2cn$ are rejected.) Now all non-rejected graphs have $m < 2cn$.  By applying \lem{searchall} we can reconstruct all edges of the graph with bounded error using $O\(\sqrt{\binom{n}{2}m}\) = O(n^{3/2})$ queries.  Given all the edges of the graph, no further queries are needed to decide $\P$.
\end{proof}

Combining this with \thm{mcsparse} and \thm{mclb}, an immediate consequence is

\begin{corollary}\label{cor:tightmc}
If $\P$ is nontrivial, minor closed, and not FSP, then $Q(\P) = \Theta(n^{3/2})$.
\end{corollary}

Note that this provides an alternative proof that the quantum query complexity of planarity is $\Theta(n^{3/2})$ \cite{AIN+08}.

For minor-closed graph properties that are also FSP, the lower bounds from \sec{LB} do not rule out the possibility of an improvement over \thm{sparse}.  In fact, we show that an improvement is possible for all such properties. 

\subsection{Quantum walk search}
\label{sec:qwalksearch}

Here we introduce our main algorithmic tool, quantum walk search.  Building on work of Ambainis \cite{Amb07} and Szegedy \cite{Sze04}, Magniez et al.\ gave the following general quantum walk search algorithm (Theorem 3 of \cite{MNRS07}):

\begin{theorem}[Quantum walk search]\label{thm:qwalksearch}
Let $P$ be a reversible, ergodic Markov chain with spectral gap $\delta>0$, and let $M$ be a subset of the states of $P$ (the \emph{marked states}) such that in the stationary distribution of $P$, the probability of choosing a marked state is at least $\epsilon>0$.  Then there is a bounded-error quantum algorithm that determines whether $M$ is empty using $O(S+\frac{1}{\sqrt\epsilon}(\frac{1}{\sqrt\delta}U+C))$ queries, where $S$ is the number of queries needed to set up a quantum sample from the stationary distribution of $P$, $U$ is the number of queries needed to update the state after each step of the chain, and $C$ is the number of queries needed to check if a state is marked.
\end{theorem}

Despite the generality of this approach, nearly all previous quantum walk search algorithms take $P$ to be a simple random walk on the Johnson graph $J(N,K)$, whose vertices are the $\binom{N}{K}$ subsets of $\{1,\ldots,N\}$ of size $K$, with an edge between subsets that differ in exactly one item.  For our purposes it will be more convenient to consider a random walk on the Hamming graph $H(N,K)$, with vertex set $\{1,\ldots,N\}^K$ and edges between two $K$-tuples that differ in exactly one coordinate.  This choice simplifies the implementation of our setup step.  Although the order of the items has no significance, and the possibility of repeated items only slows down the algorithm, the effect is not substantial.

In particular, both Markov chains have spectral gap $\delta = \Omega(1/K)$.  It can be shown that the eigenvalues of the simple random walk on the Johnson graph are $1-\frac{i(N+1-i)}{K(N-K)}$ for $i\in\{0,1,\ldots,K\}$, so the spectral gap is $\delta_{J(N,K)} = \frac{N}{K(N-K)} = \Omega(1/K)$.  The Hamming graph is even easier to analyze, since it is the Cartesian product of $K$ copies of the complete graph on $N$ vertices.  The normalized adjacency matrix of $H(N,1)$ has spectral gap $\delta_{H(N,1)}=1$, so the normalized adjacency matrix of $H(N,K)$ has spectral gap $\delta_{H(N,K)}=1/K$.

More generally, we consider a Markov chain on the tensor product of several $H(N,K_i)$ in which we take $\alpha_i$ steps on the \th{i} coordinate.  Then the spectral gap is $\delta = 1 - \max_i (1-\frac{1}{K_i})^{\alpha_i} = \Omega(\min_i {\alpha_i}/{K_i})$, where we assume that $1/K_i = o(1)$.

Note that the stationary distribution of a symmetric Markov chain is uniform.  Thus the initial state is a uniform superposition, and to calculate $\epsilon$, it suffices to calculate the probability that a uniformly random state is marked.

\subsection{Detecting subgraphs of sparse graphs}
\label{sec:subgraphsparse}

We now describe algorithms for determining whether a sparse graph $G$ contains a given subgraph $H$.  Our basic strategy is to search over subsets of the vertices of $G$ for one containing a \emph{vertex cover} of $H$, a subset $C$ of the vertices of $H$ such that each edge of $H$ involves at least one vertex from $C$.  By storing the list of neighbors of every vertex in a given subset, we can determine whether they include a vertex cover of $H$ with no further queries.  We exploit sparsity by separately considering cases where the vertices of the vertex cover have a given (approximate) degree.

Let $\vc(H)$ denote the smallest number of vertices in any vertex cover of $H$.  A vertex cover of $H$ with $\vc(H)$ vertices is called a \emph{minimal vertex cover}.

\begin{theorem}\label{thm:vcbasic}
Let $\P$ be the property that a graph either has more than $cn$ edges (for some constant $c$) or contains a given subgraph $H$.  Then $Q(\P) = \tilde O\(n^{\frac{3}{2}-\frac{1}{\vc(H)+1}}\)$.
\end{theorem}

\begin{proof}
First, we use \cor{sparsedetect} to determine whether the graph is non-sparse.  We accept if it has more than $cn$ edges.  Otherwise, we continue, knowing it has fewer than, say, $2cn$ edges.

Now let $C$ be a minimal vertex cover of $H$.  We search for a subset of the vertices of $G$ that include the vertices of $C$ (for some copy of $H$ in $G$).  To take advantage of sparsity, we separately consider different ranges for the degrees of these vertices.  We say that a vertex of $G$ has degree \emph{near} $q$ if its degree is within a constant factor of $q$.  For concreteness, let us say that the degree of $v$ is near $q$ if it is between $q/2$ and $2q$.  We search for a copy of $H$ in $G$ where vertex $i$ of $C$ has degree near $q_i$.  By considering a geometric progression of values for each of the $q_i$s, we cover all the possible vertex degrees with an overhead of only $O(\log^{\vc(H)}n) = \tilde O(1)$. 

Since \thm{qcounting} only allows us to estimate the degree of a vertex within error $\epsilon$, if the degree estimate is too close to $q/2$ or $2q$ we might incorrectly accept or reject the vertex. To handle this, we use a geometric progression where the intervals overlap enough that every possible degree is sufficiently far from the end of some interval. For concreteness, we choose the progression of values to be $2,4,8,\ldots$, so that the relevant intervals are $1$ to $4$, $2$ to $8$, $4$ to $16$, etc. 

For each fixed $(q_1,\ldots,q_{\vc(H)})$, we search over $k_i$-tuples of vertices of $G$ with degree near $q_i$ for each $i$ from $1$ to $\vc(H)$.  For each such vertex, we store its complete neighbor list.  In one step of the Markov chain, we take $\alpha_i$ steps for the \th{i} component.  Here the $k_i$s and the $\alpha_i$s are parameters that can be chosen to optimize the performance of the algorithm.

Let $t_i$ be the number of vertices of $G$ with degree near $q_i$.  Note that we can approximate the $t_i$s at negligible cost using \thm{qcounting}.  Also note that since the number of edges of $G$ is $O(n)$, we have $\sum_i t_i q_i = O(n)$, and in particular, $t_i q_i = O(n)$ for each $i$.  Choose the ordering of the indices so that $t_1 \le \cdots \le t_{\vc(H)}$.

The setup cost of the walk has two contributions.  Using Grover's algorithm, we can prepare a uniform superposition over all vertices of degree near $q_i$ using $O(n/\sqrt{t_i})$ queries.  To prepare a uniform superposition over all $k_i$-tuples of such vertices, we simply repeat this $k_i$ times, using $O(k_i n/\sqrt{t_i})$ queries.  (Note that if our search involved a Johnson graph instead of a Hamming graph, we would need to prepare a uniform superposition over $k_i$-subsets of vertices instead of $k_i$-tuples.  Although this could be done, it would make the setup step more complicated, and the performance of the algorithm would be essentially unchanged.)  Thus we can prepare a uniform superposition over the $k_i$-tuples for all $i$ using $O(\sum_i k_i n/\sqrt{t_i})$ queries.  Next we compute the list of neighbors of each of these vertices using \lem{searchall}, which takes $O(\sum_i k_i \sqrt{n q_i})$ queries.  Since $q_i = O(n/t_i)$, the cost of the neighbor computation can be neglected, so the setup cost is $S=O(\sum_i k_i n/\sqrt{t_i})$.

The cost of performing a step of the walk also has two contributions.  To update the vertices, we search for $\alpha_i$ vertices of degree near $q_i$; this takes $O(\sum_i \alpha_i n/\sqrt{t_i})$ queries.  Updating their neighbor lists takes $O(\sum_i \alpha_i \sqrt{n q_i})$ queries, which is again negligible.  Therefore, the update cost is $U=O(\sum_i \alpha_i n/\sqrt{t_i})$.  Since we perform $\poly(n)$ update steps, we reduce the error probability of each update step to $1/\poly(n)$ so that the final error probability is a constant.  This only introduces an overhead of $O(\log n)$.

We mark states of the Markov chain that include a vertex cover of a copy of $H$ in $G$.  Since we also store complete neighbor lists, and every vertex in $H$ is adjacent to some vertex of the vertex cover, no queries are required to determine whether a state is marked.  In other words, the checking cost is $C=0$.

It remains to determine the spectral gap of the chain and the fraction of marked vertices.  From \sec{qwalksearch}, the spectral gap is $\delta = \Omega(\min_i \alpha_i/k_i)$.  If we choose $k_i$ of the $t_i$ vertices of degree near $q_i$ uniformly at random, the probability of obtaining one particular vertex of degree near $q_i$ is $\Omega(k_i/t_i)$; therefore the fraction of marked vertices is $\epsilon = \Omega(\prod_i k_i/t_i)$.

Applying \thm{qwalksearch}, the number of queries used by this algorithm for any fixed $q_1,\ldots,q_{\vc(H)}$ is
\begin{align}
  O\left(n \left[\sum_i \frac{k_i}{\sqrt{t_i}} + \sqrt{\max_i \frac{k_i}{\alpha_i} \prod_i \frac{t_i}{k_i}} \sum_i \frac{\alpha_i}{\sqrt{t_i}} \right]\right).  
\end{align}
Recall that we have the freedom to choose the $\alpha_i$s and the $k_i$s.  If $\vc(H) \ge 2$, we choose them to satisfy
\begin{align}
\label{eq:choiceofvalues}
  \frac{\alpha_i}{\alpha_j} = \frac{k_i}{k_j} = \sqrt{\frac{t_i}{t_j}}
\end{align}
for all $i,j$ (which still leaves the freedom to choose one of the $\alpha_i$s and one of the $k_i$s).  Assume for now that the $\alpha_i$s and the $k_i$s are integers.
Then the query complexity is
\begin{align}
  O\left(n\left[ \frac{k_1}{\sqrt{t_1}} + \sqrt{\frac{k_1}{\alpha_1} \prod_i \frac{t_i}{k_i}} \frac{\alpha_1}{\sqrt{t_1}} \right]\right)
  &=
  O\left(n\left[ \frac{k_1}{\sqrt{t_1}} + \sqrt{\frac{k_1 \alpha_1}{t_1} \prod_i \frac{t_i}{k_i}} \right]\right) \label{eq:vcbasicqueries_start}\\
  &=
  O\left(n\left[ \frac{k_1}{\sqrt{t_1}} + \sqrt{\frac{k_1 \alpha_1}{t_1} \left(\frac{\sqrt{t_1}}{k_1}\right)^{\vc(H)} \prod_{i} \sqrt{t_i}} \right]\right) \\
  &=
  O\left(n\left[ \frac{k_1}{\sqrt{t_1}} + \sqrt{\alpha_1 \left(\frac{\sqrt{t_1 n}}{k_1}\right)^{\vc(H)-1}} \right]\right)
  \label{eq:vcbasicqueries_end}
\end{align}
where we have used the simple bound $t_i \le n$ in the last line.
Now take
$\alpha_1=1$ and
\begin{align}
\label{eq:valueofk}
  k_1 = \sqrt{t_1} n^{\frac{1}{2}-\frac{1}{\vc(H)+1}};
\end{align}
then the total query complexity is
\begin{align}
  O\(n^{\frac{3}{2} - \frac{1}{\vc(H)+1}}\)
\end{align}
as claimed.  Furthermore, recall that iterating over the various $q_i$s only introduces logarithmic overhead. Since we repeat this subroutine $\poly(\log n)$ times, we reduce the error probability of the subroutine to $1/\poly(\log n)$, which only introduces an extra $O(\log\log n) = \tilde O(1)$ factor to the query complexity.

So far we have assumed that the $\alpha_i$s and the $k_i$s are integers.  However, observe that the asymptotic expressions for the query complexity are unchanged if we replace each $\alpha_i$ by any value between $\alpha_i$ and $2\alpha_i$ and each $k_i$ by any value between $k_i$ and $2k_i$.  Since $\alpha_1=1$ is the smallest $\alpha_i$ and $k_1 = \omega(1)$ is the smallest $k_i$, and because for any $x\ge 1$ there is always an integer between $x$ and $2x$, the result holds when the $\alpha_i$s and $k_i$s are rounded up to the next largest integers.  Similarly, the fact that we only have a multiplicative approximation for the $t_i$s does not affect the asymptotic running time.
\end{proof}

We can apply this algorithm to decide sparse graph properties, and in particular minor-closed properties, that are also FSP:  we simply search for each of the forbidden subgraphs, accepting if none of them are present.  For minor-closed properties, the non-sparseness condition of \thm{vcbasic} can be removed due to \thm{mcsparse}.  Thus, since $\vc(H)$ is a constant for any fixed graph $H$, we have

\begin{corollary}\label{cor:sparsefsp}
If $\P$ is sparse and FSP, then $Q(\P) = O(n^\alpha)$ for some $\alpha < 3/2$.
\end{corollary}

Note that \thm{vcbasic} also holds if we ask whether $H$ is contained as an induced subgraph, since when we check whether $H$ is present for a certain subset of edges, we have access to their complete neighbor lists.

For many subgraphs, we can improve \thm{vcbasic} further by storing additional information about the vertices in the minimal vertex cover: in addition to storing their neighborhoods, we can also store basic information about their second neighbors.  In particular, we have the following.

\begin{theorem}\label{thm:vcdangling}
Let $\P$ be the property that a graph either has more than $cn$ edges (for some constant $c$) or contains a given subgraph $H$.  Let $H'$ be the graph obtained by deleting all degree-one vertices of $H$ that are not part of an isolated edge.  Then $Q(\P) = \tilde O\(n^{\frac{3}{2}-\frac{1}{\vc(H')+1}}\)$.
\end{theorem}

\begin{proof}
As before, we begin by using \cor{sparsedetect} to determine whether the graph is non-sparse, accepting if this is the case.

Otherwise, we use the technique of color-coding \cite{AYZ95} to handle the degree-one vertices of $H$.
Suppose that $H$ has $\ell$ degree-one vertices, and label them $1,\ldots,\ell$; label the other vertices $\ell+1$.  Assign labels from the set $\{1,\ldots,\ell,\ell+1\}$ uniformly at random to the vertices of $G$.  If there is a copy of $H$ in $G$, then with probability at least $(\ell+1)^{-|V(H)|} = \Omega(1)$, the vertices of this copy of $H$ in $G$ have the correct labels.  We assume this is the case, increasing the cost of the algorithm by a factor of $O(1)$.

We augment the algorithm of \thm{vcbasic} by storing additional information about each vertex: in addition to the neighborhood, we also store whether each vertex has a second neighbor with each possible label.  Computing this information for any one vertex of $G$ of degree near $q_i$ with known neighborhood takes $O(\sqrt{n q_i})$ queries, so the additional setup cost of $O(\sum_i k_i \sqrt{n q_i})$ and update cost of $O(\sum_i \alpha_i \sqrt{n q_i})$ to store this information is negligible.  Now we can recognize $H$ by storing only a minimal vertex cover of $H'$, still with zero checking cost.  Thus the same analysis applies with $\vc(H)$ replaced by $\vc(H')$, and the result follows.
\end{proof}

In particular, \thm{vcdangling} gives an improvement over \thm{vcbasic} for the properties that are characterized by a single forbidden minor (and equivalently, a single forbidden subgraph).  For such properties, we have the following:

\begin{proposition}\label{prop:vcpathclaw}
If $H$ is a $k$-path then $\vc(H) = \ceil{k/2}$.  If $H$ is a $\{d_1,d_2,d_3\}$-claw with $d_1,d_2,d_3$ all even, then $\vc(H) = (d_1+d_2+d_3)/2$; if at least one of $d_1,d_2,d_3$ is odd, then $\vc(H) = 1+\sum_{i=1}^3 \ceil{\frac{d_i-1}{2}}$.
\end{proposition}

Therefore, \thm{vcdangling} implies

\begin{corollary}\label{cor:pathfinding}
A $k$-path with $k \ge 3$ can be detected using $\tilde O(n^{\frac{3}{2}-\frac{1}{\ceil{k/2}}})$ quantum queries.  In particular, the quantum query complexity of detecting a $k$-path for $k \in \{1,2,3,4\}$ is $\tilde \Theta(n)$.  Furthermore, a $\{d_1,d_2,d_3\}$-claw with $d_1+d_2+d_3 > 3$ can be detected using $\tilde O(n^{\frac{3}{2}-\frac{2}{d_1+d_2+d_3-1}})$ queries if $d_1,d_2,d_3$ are all odd, and with $\tilde O(n^{\frac{3}{2}-(\sum_{i=1}^3 \ceil{d_i/2}-1)^{-1}})$ queries otherwise.
\end{corollary}

Note that even the improved result from \thm{vcdangling} has zero checking cost.  We can sometimes obtain a further improvement by performing nontrivial checking.  For example, we can detect $7$-paths in the same complexity that \cor{pathfinding} gives for $5$- and $6$-paths and we can detect $9$- and $10$-paths in the same complexity that \cor{pathfinding} gives for $7$- and $8$-paths.

\begin{theorem}\label{thm:paths7-10}
$H$-subgraph containment has query complexity $\tilde O(n^{7/6})$ if $H$ is a $7$-path and $\tilde O(n^{5/4})$ if $H$ is a $9$- or $10$-path. Furthermore, for $k>10$, the query complexity of finding $k$-paths is $\tilde O(n^{\frac{3}{2} - \frac{1}{\ceil{k/2}+1}})$.
\end{theorem}

\begin{proof}
First we describe the algorithm for 7-paths. As before, the first step is to accept graphs that are sufficiently dense. Then, following \thm{vcdangling}, let $H'$ be the graph with the degree-1 vertices deleted, namely, a 5-path. The minimal vertex cover of a 5-path contains 3 vertices. Instead we store only two vertices: the second and fifth vertex of the 5-path. Since these vertices are exactly one vertex from the end of the 5-path, the color-coding trick from \thm{vcdangling} allows us to find the deleted vertices.

We store the second and fifth vertex of the 5-path together with complete information about their neighborhoods. However, this does not allow us to determine if the two vertices are connected by a 3-path, which is what we require. We test this condition during the checking step, adding a nontrivial checking cost to the algorithm. Since we know all the neighbors of these vertices, we only have to check if there is an edge between a neighbor of the second vertex and a neighbor of the fifth vertex to determine if there is a 3-path between the vertices. As before, we store $k_1$ vertices of degree near $q_1$ and $k_2$ vertices of degree near $q_2$. These vertices can have at most $k_1q_1$ and $k_2q_2$ neighbors respectively. Searching for an edge over these possibilities requires $C=O(\sqrt{k_1q_1k_2q_2})$ queries.

According to \thm{qwalksearch}, the query complexity overhead due to the checking step is $C/\sqrt{\epsilon}$. Since $\epsilon=k_1/t_1 \times k_2/t_2$, the overhead is $O(\sqrt{t_1q_1t_2q_2})$, which is $O(n)$ since $t_iq_i = O(n)$. Thus the total query complexity is $O(n^{\frac{3}{2}-\frac{1}{2+1}})+O(n) = O(n^{7/6})$.

The algorithm for 10-paths (which immediately implies an algorithm for 9-paths) is slightly different. First, we accept graphs that are sufficiently dense. Then we delete the degree-1 vertices, which leaves us with an 8-path. We store the second, fifth, and eighth vertices of the path. As before, the deleted vertices are detected using the color-coding trick.

Suppose the second, fifth and eighth vertices have degrees near $q_1$, $q_2$, and $q_3$, respectively. We have to check if there is a 3-path connecting the second vertex to the fifth vertex, and another connecting the fifth vertex to the eighth vertex. For the checking step, we first fix an arbitrary vertex as the fifth one (there are $k_2$ such vertices that can be fixed), and then search over the neighbors of this vertex for a neighbor of the second and eighth vertex. Since the neighbors of these vertices are known, this costs $O(\sqrt{q_2 k_1q_1} + \sqrt{q_2 k_3q_3})$. Now we search over the possible fifth vertices to see if any of the $k_2$ possibilities works, giving a total checking cost of  $C=O(\sqrt{k_2q_2}(\sqrt{k_1q_1} + \sqrt{k_3q_3}))$. Since this contributes $C/\sqrt{\epsilon}$ to the total query complexity, and  $\epsilon=k_1k_2k_3/t_1t_2t_3$, we get $C/\sqrt{\epsilon}= O(n(\sqrt{t_1/k_1} + \sqrt{t_3/k_3}))$, where we have used the fact that $q_it_i = O(n)$. Since \eq{choiceofvalues} and \eq{valueofk} imply that $k_i = \sqrt{t_i} n^{1/4}$, and $t_i=O(n)$, we have $C/\sqrt{\epsilon} = O(n^{9/8}) \le O(n^{5/4})$, giving the claimed total query complexity of $\tilde O(n^{5/4})$.

Finally, for paths of length greater than 10, we use a similar approach. First, we delete the degree-1 vertices of the $k$-path to get a $(k-2)$-path. We then store the second, fifth, eighth, and all subsequent even numbered vertices. Thus, for a $k$-path, we store $l = \ceil{k/2} - 2$ vertices. The degree-1 vertices are handled using the color-coding trick, and since we store the neighbors of all the selected vertices, the vertices on the path beyond the eighth vertex are known. We only have to check if there is a 3-path connecting the second vertex to the fifth vertex, and the fifth vertex to the eighth vertex as before. The checking cost remains the same as before, $C=O(\sqrt{k_2q_2}(\sqrt{k_1q_1} + \sqrt{k_3q_3}))$. However, the value of epsilon is now $\epsilon=\prod_{i=1}^l k_i/t_i$. Using \eq{choiceofvalues}, \eq{valueofk}, and $t_i=O(n)$ as before, we get 
\begin{equation}
\frac{C}{\sqrt{\epsilon}} = \sqrt{\frac{\prod_i t_i}{\prod_i k_i}}\(\sqrt{k_2q_2}\(\sqrt{k_1q_1} + \sqrt{k_3q_3}\)\)=\frac{n^{\frac{l-2}{4}}}{(n^{\frac{1}{2} - \frac{1}{l+1}})^{\frac{l-2}{2}}} \times O(n)=O(n^{1+\frac{l-2}{2l+2}}),
\end{equation}
which is $O(n^{\frac{3}{2}-\frac{1}{l+1}})$ for all $l$. Thus the update and setup costs dominate, which gives us an $\tilde O(n^{\frac{3}{2}-\frac{1}{\ceil{k/2}-1}})$ query algorithm for $k$-paths when $k>10$.
\end{proof}

The following table summarizes the best upper bounds presented for detecting a path of length $k$. Note that the upper bounds are monotone in $k$ since finding a $(k+1)$-path is at least as hard as finding a $k$-path.

\begin{table}[h]
\capstart
\begin{center}
    \begin{tabular}{ | c | c | c | c | c | c  | c | c | }
    \hline
    \centering{$k$} & 1 -- 4 & 5 & 6 & 7 & 8 & 9 & $k\geq 10$ \\ \hline
    Upper bound & $\tilde \Theta(n)$ & $\tilde O(n^{7/6})$ & $\tilde O(n^{7/6})$ & $\tilde O(n^{7/6})$ & 
	$\tilde O(n^{5/4})$ & $\tilde O(n^{5/4})$ & 
	$\tilde O\(n^{\frac{3}{2}-\frac{1}{\ceil{k/2}-1}}\)$\\
    \hline
    \end{tabular}
\end{center}
\caption{A summary of our algorithms for the $k$-path containment problem.
\label{tab:paths}}
\end{table}

\subsection{Relaxing sparsity}

In the previous section, we focused on the case of sparse graphs, since this is the relevant case for minor-closed graph properties.  However, our algorithms easily generalize to the case where the number of edges is at most some prescribed upper bound, which need not be linear in $n$, leading to further applications. 

\begin{theorem}\label{thm:vcpseudosparse}
Let $\P$ be the property that an $n$-vertex graph either has more than $\bar m$ edges, where $\bar m = \Omega(n)$, or contains a given subgraph $H$.  Let $H'$ be the graph obtained by deleting all degree-one vertices of $H$ that are not part of an isolated edge.  Then $Q(\P) = \tilde O(\sqrt{\bar m} n^{1-\frac{1}{\vc(H')+1}})$.
\end{theorem}

Note that \thm{vcpseudosparse} subsumes \thm{vcdangling}, which in turn subsumes \thm{vcbasic}.

\begin{proof}
  We apply the algorithm from the proof of \thm{vcdangling} (which in turn depends on the analysis in the proof of \thm{vcbasic}), replacing the promise that $m = O(n)$ with the promise that $m = O(\bar m)$, which implies that $q_i = O(\bar m/t_i)$.  In this algorithm, the setup cost is
\begin{align}
  S &= O\left(\sum_i k_i n/\sqrt{t_i} + \sum_i k_i \sqrt{n q_i}\right) \\
    &= O\left(\sum_i \frac{k_i}{\sqrt{t_i}}(n + \sqrt{n \bar m})\right) \\
    &= O\left(\sqrt{n \bar m}\sum_i \frac{k_i}{\sqrt{t_i}}\right),
\end{align}
and by a similar calculation, the update cost is
\begin{align}
  U
  &= O\left(\sqrt{n \bar m} \sum_i \frac{\alpha_i}{\sqrt{t_i}}\right);
\end{align}
the checking cost remains $C=0$, and the spectral gap and fraction of marked vertices are also unchanged.  Again choosing
\begin{align}
  \frac{\alpha_i}{\alpha_j} = \frac{k_i}{k_j} = \sqrt{\frac{t_i}{t_j}},
\end{align}
the query complexity from \thm{qwalksearch} is
\begin{align}
  O\left(\sqrt{n \bar m}\left[ \frac{k_1}{\sqrt{t_1}} + \sqrt{\frac{k_1}{\alpha_1} \prod_i \frac{t_i}{k_i}} \frac{\alpha_1}{\sqrt{t_1}} \right]\right)
  &=
  O\left(\sqrt{n \bar m}\left[ \frac{k_1}{\sqrt{t_1}} + \sqrt{\alpha_1 \left(\frac{\sqrt{t_1 n}}{k_1}\right)^{\vc(H)-1}} \right]\right).
\end{align}
(cf.~\eqrange{vcbasicqueries_start}{vcbasicqueries_end}).  Taking
$\alpha_1=1$ and
\begin{align}
  k_1 = \sqrt{t_1} n^{\frac{1}{2}-\frac{1}{\vc(H')+1}}
\end{align}
gives a query complexity of
\begin{align}
  O(\sqrt{\bar m}n^{1-\frac{1}{\vc(H')+1}})
\end{align}
and the result follows.
\end{proof}

In conjunction with the K\"{o}v\'{a}ri--S\'{o}s--Tur\'{a}n theorem~\cite{KST54}, this algorithm has applications to subgraph-finding problems that are not equivalent to minor-finding problems. 

\begin{theorem}[K\"{o}v\'{a}ri--S\'{o}s--Tur\'{a}n]
\label{thm:KST}
If a graph $G$ on $n$ vertices does not contain $K_{s,t}$ as a subgraph, where $1 \leq s \leq t$, then $|E(G)| \leq c_{s,t} \, n^{2-\frac{1}{s}}$, where $c_{s,t}$ is a constant depending only on $s$ and $t$.
\end{theorem}

In particular, we use this theorem to show the following.

\begin{theorem}\label{thm:bipartite}
If $H$ is a $d$-vertex bipartite graph, then $H$-subgraph containment has quantum query complexity $\tilde O(n^{2-\frac{1}{d}-\frac{2}{d+2}}) = \tilde O(n^{2-\frac{3d+2}{d(d+2)}})$.
\end{theorem}

\begin{proof}
By \thm{KST}, a graph that does not contain $K_{s,t}$ (where $1 \le s \le t$) has at most $c_{s,t} \, n^{2-\frac{1}{s}}$ edges, so a graph with more than $c_{s,t} \, n^{2-\frac{1}{s}}$ edges must contain $H$.  \thm{qcounting} shows that we can determine whether the input graph has more than $2c_{s,t} \, n^{2-\frac{1}{s}}$ edges using $o(n)$ queries.  If so, we accept; otherwise, we apply \thm{vcpseudosparse} with $\bar m = 2c_{s,t} \, n^{2-\frac{1}{s}}$ and $\vc(H') \le d/2$, giving the desired result.
\end{proof}

Recall that for $d>3$, Theorem 4.6 of \cite{MSS07} gives an upper bound of $O(n^{2-\frac{2}{d}})$ for finding a $d$-vertex subgraph.  For bipartite subgraphs, \thm{bipartite} is a strict improvement.

Note that a better bound may be possible by taking the structure of $H$ into account.  In general, if $H$ is a bipartite graph with the \th{i} connected component having vertex bipartition $V_i \cup U_i$ with $1 \le |V_i| \le |U_i|$, then we can replace $d/2$ by $\sum_i |V_i|$, since a graph that does not contain $K_{\sum_i |V_i|,\sum_i |U_i|}$ does not contain $H$, and $\vc(H') \le \vc(H) = \sum_i |V_i|$.  As a simple example, if $H=K_{1,t}$ is a star on $t+1$ vertices, then $H$-subgraph containment can be solved with $\tilde O(n)$ quantum queries (which is essentially optimal due to \thm{subgraphlb}).  This shows that the quantum query complexity of deciding if a graph is $t$-regular is $\tilde\Theta(n)$.  (In fact it is not hard to show that the query complexity of this problem is $\Theta(n)$.)

As another example, consider the property of containing a fixed even-length cycle, i.e., $C_{2l}$-subgraph containment.  Since $C_{2l}$ is bipartite, this is a special case of the problem considered above.  \thm{bipartite} gives an upper bound of $\tilde O(n^{2-\frac{3l+1}{2l(l+1)}})$, which approaches $O(n^2)$ as the cycle gets longer (i.e., as $l \rightarrow \infty$). As concrete examples, it gives upper bounds of $\tilde O(n^{1.41\overline{6}})$ for $C_4$ containment and $\tilde O(n^{1.58\overline{3}})$ for $C_6$ containment.

For even cycles of length 6 or greater, this estimate can be significantly improved by replacing \thm{KST} with the following result of Bondy and Simonovits~\cite{BS74}.

\begin{theorem}[Bondy--Simonovits]
\label{thm:BS}
Let $G$ be a graph on $n$ vertices. For any $l \geq 1$, if $|E(G)| > 100ln^{1+1/l}$ then $G$ contains $C_{2l}$ as a subgraph.
\end{theorem}

Using this upper bound instead of \thm{KST} in \thm{bipartite} gives us the following upper bound for even cycles.

\begin{theorem}
\label{thm:evencycles}
The $C_{2l}$-subgraph containment problem can be solved using $\tilde O(n^{\frac{3}{2}-\frac{l-1}{2l(l+1)}})$ queries.
\end{theorem}

For $C_4$ containment, the upper bound given by this theorem matches the one given by \thm{bipartite}. However, for all longer even cycles, the bound given by this theorem is strictly better than the one given by \thm{bipartite}. For example, we get an upper bound of $\tilde O(n^{1.41\overline{6}})$ for $C_6$ containment, as compared to the upper bound of  $\tilde O(n^{1.58\overline{3}})$ given by \thm{bipartite}. Moreover, as the cycles get longer, the upper bound of \thm{evencycles} approaches $O(n^{3/2})$ instead of $O(n^2)$.

As in \thm{paths7-10}, we can sometimes improve over \thm{vcpseudosparse} by introducing a nontrivial checking cost.  The following is a simple example of such an algorithm for $C_4$ containment that performs better than the  $\tilde O(n^{1.41\overline{6}})$ query algorithm given by both \thm{bipartite} and \thm{evencycles}.

\begin{theorem}\label{thm:fourcycle}
$C_4$-subgraph containment can be solved in $\tilde O(n^{1.25})$ quantum queries.
\end{theorem}

This may seem unexpected, since $C_4$ finding is a natural generalization of triangle finding to a larger subgraph.  Indeed, the previous best known quantum algorithm for $C_4$ finding used $\tilde O(n^{1.5})$ queries \cite{MSS07}, more than the $O(n^{1.3})$ queries for triangle finding.  Our improvement shows that $4$-cycles can be found in fewer quantum queries than in the best known quantum algorithm for finding $3$-cycles.

\begin{proof}
By \thm{KST}, a graph with $\Omega(n^{3/2})$ edges must contain $C_4 = K_{2,2}$, so after applying \thm{qcounting} and accepting graphs with $\Omega(n^{3/2})$ edges, we can assume that $m = O(n^{3/2})$.  Then, for various values of $q$ in a geometric progression, we search for a vertex $v$ of the input graph that has degree near $q$ and that belongs to a $4$-cycle.
It suffices to search for $v$ using Grover's algorithm instead of the more general \thm{qwalksearch}.  Let $t$ denote the number of vertices of the input graph with degree near $q$, and let $\ket{\psi}$ denote a uniform superposition over all such vertices, where for each vertex we store a list of its $O(q)$ neighbors.  Grover's algorithm starts from the state $\ket{\psi}$ and alternates between reflecting about $\ket{\psi}$ and about vertices that belong to a $4$-cycle, detecting whether such a vertex exists in $O(\sqrt{t})$ iterations.  By Grover's algorithm, we can prepare a uniform superposition over the vertices of degree near $q$ using $O(n/\sqrt t)$ queries; by \lem{searchall}, we can compute the neighborhood of a vertex with degree near $q$ in $O(\sqrt{nq})$ queries.  Thus we can prepare or reflect about $\ket{\psi}$ in $O(n/\sqrt t + \sqrt{nq})$ queries.  Given the neighbors of a vertex with degree near $q$, we can decide whether that vertex is part of a $4$-cycle by searching over all vertices in the graph for a vertex that is adjacent to 2 of its neighbors.  This can be done in $O(\sqrt{nq})$ queries.  Thus the query complexity of searching for a $4$-cycle is
\begin{align}
  O\left( \sqrt{t} \left[ \frac{n}{\sqrt t} + \sqrt{nq} \right] \right)
  &= O( n + \sqrt{nqt} ).
\end{align}
Since $qt = O(n^{3/2})$, we see that the number of amplitude amplification steps is $O(n^{5/4})$.  As in previous algorithms, iterating over values of $q$ and error reduction of subroutines only introduces logarithmic overhead, so the total query complexity is $\tilde O(n^{5/4})$ as claimed.
\end{proof}

\section{Conclusions and open problems}
\label{sec:open}

In this paper, we have studied the quantum query complexity of minor-closed graph properties.  The difficulty of such problems depends crucially on whether the property can also be characterized by a finite set of forbidden subgraphs.  Minor-closed properties that are not characterized by forbidden subgraphs have matching upper and lower bounds of $\Theta(n^{3/2})$ (\cor{tightmc}), whereas all minor-closed properties that can be expressed in terms of forbidden subgraphs can be solved strictly faster, in $O(n^\alpha)$ queries for some $\alpha<3/2$ that may depend on the property (\cor{sparsefsp}).

Since the best known lower bound for the latter class of problems is the simple $\Omega(n)$ lower bound from \thm{subgraphlb}, an obvious open question is to give improved upper or lower bounds for subgraph-finding problems.  While the standard quantum adversary method cannot prove a better lower bound, it might be possible to apply the negative weights adversary method \cite{HLS07} or the polynomial method~\cite{BBC+01}.  Note that sparsity makes forbidden subgraph properties potentially more difficult to lower bound; this is precisely the feature we took advantage of in the algorithms of \sec{algo}.  Proving a superlinear lower bound for any subgraph-finding problem---even one for which dense graphs might not contain the subgraph, such as in the case of triangles---remains a major challenge.  On the algorithmic side, note that while our algorithms take advantage of sparsity, minor-closed families of graphs have other special properties, such as bounded degeneracy, which might also be exploited.

The algorithms described in \sec{algo} have several features not shared by previous quantum walk search algorithms for graph properties: queries are required even to identify which vertices of the input graph to search over (namely, to find vertices of a certain degree), and the performance of the walk is optimized by making different transitions at different rates.  We hope these techniques might prove useful in other quantum algorithms.

We observed that \thm{vcbasic} can be applied to find induced subgraphs (just as with the algorithms of \cite{MSS07}).  However, the improvements described in \thm{vcdangling} and \thm{vcpseudosparse} do not apply to induced subgraphs, and in general it could be easier or more difficult to decide whether a given graph is present as an induced subgraph rather than a (not necessarily induced) subgraph.  It might be fruitful to explore induced subgraph finding more generally.

It might also be interesting to focus on finding natural families of subgraphs such as paths.  Recall that we showed the quantum query complexity of this problem is $\tilde \Theta(n)$ for lengths up to $4$ and $\tilde O(n^{7/6})$ for lengths of $5$, $6$, and $7$, with nontrivial algorithms for longer paths as well (\cor{pathfinding} and \thm{paths7-10}, as summarized in \tab{paths}).  The case of paths of length $5$, the smallest case for which our algorithm is not known to be optimal, appears to be a natural target for future work.


\section*{Acknowledgments}

We thank Fr\'ed\'eric Magniez for suggesting that a strategy similar to the lower bound for connectivity in \cite{DHHM06} could be applied to $C_3$-minor finding.  We also thank Nathann Cohen for pointing us to Ref.~\cite{RS90} (via the website \href{http://www.mathoverflow.net}{www.mathoverflow.net}). R.\ K. thanks Jim Geelen for helpful discussions about graph minors, and Vinayak Pathak and Hrushikesh Tilak for interesting conversations about many aspects of this work.
This work was supported in part by MITACS, NSERC, QuantumWorks, and the US ARO/DTO.


\bibliographystyle{plain}
\bibliography{minor}

\end{document}